\documentclass{article}

\PassOptionsToPackage{numbers,sort&compress}{natbib}
\makeatletter
\def\input@path{{./26.5-NeurIPS2026/}{./}{../}}%
\makeatother
\usepackage[preprint]{neurips_2026}

\usepackage[utf8]{inputenc}
\usepackage[T1]{fontenc}
\usepackage{hyperref}
\usepackage{url}
\usepackage{nicefrac}
\usepackage{microtype}

\usepackage{adjustbox}
\usepackage{algorithmic}
\usepackage[ruled,linesnumbered]{algorithm2e} 

\usepackage{amsmath,amsfonts,amsthm} 

\usepackage{booktabs} 
\usepackage{bm} 
\usepackage{diagbox} 
\usepackage{float} 
\usepackage{graphicx} 

\usepackage{listings} 
\usepackage{mathrsfs} 
\usepackage{mathtools} 
\usepackage{multirow} 

\usepackage[sort&compress, numbers]{natbib}

\usepackage{subcaption} 
\usepackage{thmtools} 
\usepackage{thm-restate} 
\usepackage{tikz} 

\usepackage[textsize=normalsize]{todonotes} 
\usepackage{xcolor} 

\usepackage[capitalize,noabbrev]{cleveref} 

\definecolor{pku-red}{RGB}{139,0,18} 

\theoremstyle{plain} 
\newtheorem{theorem}{Theorem}[section] 
 
\newtheorem{lemma}[theorem]{Lemma} 

\newtheorem{example}[theorem]{Example}

\theoremstyle{definition}
\newtheorem{definition}[theorem]{Definition}
\newtheorem{assumption}[theorem]{Assumption}
\newtheorem{fact}[theorem]{Fact}
\theoremstyle{remark}
\newtheorem{remark}[theorem]{Remark}

\crefname{assumption}{Assumption}{Assumptions}
\Crefname{assumption}{Assumption}{Assumptions}
\crefname{algorithm}{Algorithm}{Algorithms}
\Crefname{algorithm}{Algorithm}{Algorithms}
\crefname{appendix}{Appendix}{Appendices}
\Crefname{appendix}{Appendix}{Appendices}

\newcommand{\bbE}{\mathbb{E}}

\newcommand{\bbR}{\mathbb{R}}

\newcommand{\calF}{\mathcal{F}}

\newcommand{\calL}{\mathcal{L}}

\newcommand{\calP}{\mathcal{P}}

\newcommand{\calV}{\mathcal{V}}

\newcommand{\bmp}{{\bm{p}}}

\newcommand{\bms}{{\bm{s}}}
\newcommand{\bmt}{{\bm{t}}}

\newcommand{\bmv}{{\bm{v}}}

\newcommand{\bmx}{{\bm{x}}}

\newcommand{\iinn}{{i\in[n]}} 
\newcommand{\jinn}{{j\in[n]}}

\newcommand{\kinL}{{k\in[L]}}
\newcommand{\linL}{{l\in[L]}}

\newcommand{\jlei}{{j\leq i}}

\newcommand{\jgei}{{j\geq i}}
\newcommand{\jggei}{{j > i}}

\newcommand{\one}{\bm{1}} 

\newcommand{\argmax}{\mathrm{argmax}}

\newcommand{\dd}{\mathrm{d}} 
\newcommand{\poly}{\mathrm{poly}} 

\newcommand{\ie}{\emph{i.e.}}    
\newcommand{\eg}{\emph{e.g.}}    
\newcommand{\wlg}{\emph{w.l.o.g.}}    
\newcommand{\resp}{\emph{resp.}} 

\newcommand{\Wlg}{\wlg~}

\newcommand{\SW}{\mathrm{SW}}

\newcommand{\tbmp}{\tilde{\bmp}}

\newcommand{\tbmx}{\tilde{\bmx}}
\newcommand{\tbms}{\tilde{\bms}}

\newcommand{\tp}{\tilde{p}}

\newcommand{\tx}{\tilde{x}}

\newcommand{\ts}{\tilde{s}}

\newcommand{\tmu}{{\tilde{\mu}}}
\newcommand{\bmmu}{{\bm{\mu}}}
\newcommand{\tbmmu}{{\tilde{\bmmu}}}

\usepackage[normalem]{ulem} 
\usepackage{marginnote}
\let\oldcite\cite
\renewcommand{\cite}[1]{\mbox{\oldcite{#1}}}

\usepackage{enumerate}
\usepackage[inline]{enumitem}

\title{Incentivizing Data Trading via Profit Reallocation}

\author{%
Yunxuan Ma\thanks{Equal contribution.} \\
Peking University \\
\texttt{yunxuanma@pku.edu.cn} \\
\And
Wu Xin\footnotemark[1] \\
Peking University \\
\texttt{xinwu@pku.edu.cn} \\
\And
Jichen Li \\
Tsinghua University \\
\texttt{jichenli@mail.tsinghua.edu.cn} \\
\And
Hongyin Chen \\
Technion \\
\texttt{hongyin.chen.contact@gmail.com} \\
\And
Xiaoqi Dong \\
Renmin University of China \\
\texttt{xiaoqi.dong@ruc.edu.cn} \\
\And
Yusen Zheng \\
Peking University \\
\texttt{yusen@stu.pku.edu.cn} \\
\And
Yukun Cheng \\
Jiangnan University \\
\texttt{ykcheng@amss.ac.cn} \\
\And
Xiaotie Deng\thanks{Corresponding author.} \\
Peking University \\
\texttt{xiaotie@pku.edu.cn} \\
}

\begin{document}

\maketitle
\begin{abstract}


Data trading is a central approach to data circulation, yet data markets remain far less active than expected.
A primary bottleneck is the lack of effective economic incentives.
Existing approaches often treat data as traditional goods, overlooking its inherent replicability and resale potential: buyers can replicate and resell data products, thereby forming transaction chains.
Upstream sellers do not benefit from downstream resales and thus have limited incentives to sell.
However, the impact of data resale on market performance remains insufficiently understood.

To address this gap, we propose a sequential, chain-based data trading model that explicitly captures data resale.
The model reflects data flows in settings such as LLM training and strategic decision-making.
We integrate this model with a profit reallocation mechanism.
By reallocating profits along the transaction chain, this mechanism ensures upstream sellers benefit from downstream resales.
We next develop efficient algorithms, including a polynomial-time exact algorithm for the discrete model and an FPTAS for the continuous model, to compute its sequential equilibria.
We theoretically show that profit reallocation expands trade and improves social welfare under certain conditions,
and empirical results demonstrate that our mechanism increases transaction volume by 120.0\% and social welfare by 50.4\% in synthetic environments, compared with the baseline mechanism that does not reallocate profits.
\end{abstract}

\section{Introduction}
\label{sec:intro}

Data trading is a key approach to facilitating data circulation and unlocking data value \cite{Fernandez2020DataMarketPlatforms,agarwal2019marketplace,biswas2021incentive}, with industrial implementations including Snowflake Data Marketplace \cite{SnowflakeDocs_MarketplaceAbout}, AWS Data Exchange \cite{AWSDocs_DataExchangeLanding}, and Dawex \cite{Dawex_Web_DataMarketplaceModel}.
However, data markets are insufficiently active and remain underdeveloped \citep{Koutroumpis2017UnfulfilledPotential,OECD2021ValueOfData,EuropeanParliament2023DataActPressRelease}.
A key reason is the lack of effective incentive mechanisms for market participants \citep{OECD_EASD_2019,azcoitia2022survey, bauer2024designing,xin2025tbds}.

Existing research has sought to address this challenge through data auctions \citep{ghosh2011selling,zhang2020selling,agarwal2024towards,ravindranath2023data},
profit sharing \citep{cao2017game,yu2020fairness},
marketplace design \citep{agarwal2019marketplace,biswas2021incentive},
and incentive mechanisms \citep{li2023martfl,zhang2025incentive}.
However, these approaches often overlook the replicability and resale potential of data~\citep{OECD2021ValueOfData,JonesTonetti2020NonrivalryData}.
Unlike traditional goods, data can be replicated, processed into data products, and resold, thereby forming transaction chains \citep{xin2025tbds,azcoitia2022survey,biswas2021incentive}.
Because upstream sellers do not benefit from downstream resales, they may have weak incentives to sell data in the first place.
Some studies recognize that data can generate profit through resale~\citep{biswas2021incentive,xin2025tbds}.
However, \citet{biswas2021incentive} does not study incentive design for data resale, and \citet{xin2025tbds}'s approach relies on public information.
To the best of our knowledge, no prior work has addressed incentive design in data markets while explicitly accounting for data resale.

To bridge this gap, we study incentive design for data trading when resale is non-negligible, and model the system as a sequential, chain-based data trading game that explicitly captures resale behavior.

Building upon this model, we propose a profit reallocation mechanism that ensures upstream sellers continue to benefit from downstream trades, thereby incentivizing their willingness to trade data.
An important contribution is that our mechanism generalizes fixed-ratio profit reallocation \citep{xin2025tbds,Swash2021Whitepaper}, which has limited expressiveness. We only require the mechanism to be budget feasible, a general and economically meaningful condition.

We develop two efficient algorithms to compute sequential equilibria under the proposed mechanism: a polynomial-time exact algorithm for discrete models and a fully polynomial-time approximation scheme (FPTAS) for continuous models.
These equilibrium computation algorithms enable efficient evaluation of profit reallocation mechanisms.
The core algorithmic challenge is to compute player utilities efficiently, as naive computation has exponential complexity.
We design a two-level dynamic program to enable polynomial computation, which might be of independent interest.

Finally, we theoretically establish the advantages of profit reallocation: under certain conditions, compared with the baseline mechanism that does not reallocate profits, profit reallocation expands the equilibrium trade event and improves social welfare. This comparison is non-trivial because changes in the mechanism can affect the resulting equilibria in complex ways.
Our theoretical findings are quantified through experiments: compared with the baseline mechanism, profit reallocation mechanism increases transaction volume by $120.0\%$ and social welfare by $50.4\%$ in simulated environments.

In summary, our contributions are: (1) a chain-based data trading game that models participants' strategic behavior (\cref{subsec:game,subsec:solution}), (2) a generalized profit reallocation mechanism (\cref{subsec:mechanism}), (3) efficient equilibrium computation algorithms (\cref{sec:algo}), and (4) theoretical (\cref{sec:theoretic}) and experimental (\cref{sec:exp}) validation of profit reallocation mechanisms.

Our approach applies directly to numerous chain-based data trading settings and suggests extensions to richer data trading topologies, such as trees or directed acyclic graphs (DAGs), that build on chains\footnote{For space limits, further related works are deferred to \cref{sec:related}.}.

\section{Model}
\label{sec:model}

\subsection{Sequential Data Trading Game}\label{subsec:game}

We consider a sequential data trading game on a data platform, involving $n + 1$ players indexed by $i \in \{0, 1, \dots, n\}$. In this setting, replicable data are traded along a chain. After acquiring data from player $i-1$, each player $i \in [n] \coloneqq \{1,\cdots,n\}$ can process the data into a usable ``data product'', which can then be sold downstream to player $i+1$. Player $i$ may obtain value from the data product and incur a processing cost; we use $v_i \in \mathcal{V}_i \subseteq \mathbb{R}$ to denote the value minus the processing cost, which may be positive or negative. The variable $v_i$ is player $i$'s private information, and we refer to it as \emph{player $i$'s valuation}. We assume that valuations are Markovian: $v_i$ depends only on the preceding valuation $v_{i-1}$.

\begin{assumption}[Markovian] 
\label{def:markov}
For each player $i \in [n]$, the valuation $v_i$ depends only on the immediate preceding valuation $v_{i-1}$. Formally, for any $i \geq 1$, the conditional distribution of $v_i$ satisfies:
$f_i(v_i \mid v_0, v_1, \dots, v_{i-1}) = f_i(v_i \mid v_{i-1})$,
where $f_i$ denotes the conditional PDF (or $F_i$ for the conditional CDF and $\calF$ for the joint distribution).
\end{assumption}

This Markovian structure is natural for two reasons: (a) it captures the intuition that if player $i$ creates a data product $P$ from an input $I$, then the value of $P$ should depend on the characteristics of $I$ rather than on the more ``atomic'' data used to generate $I$; and (b) it strictly relaxes the standard assumption of independent private valuations \citep{myerson1981optimal}. We assume that $\calF$, $f_i$s, and $F_i$s are common knowledge, as is standard in the economics literature \citep{MWG}.

The game proceeds through $n$ sequential trades. In trade $i$, player $i-1$ acts as the seller and makes a take-it-or-leave-it offer by posting a non-negative price $p_{i-1} \in \mathbb{R}_+$ on the platform.\footnote{We require the price to be non-negative because a negative price is economically meaningless.} After observing the posted price, buyer $i$ decides the willingness to buy ($x_i = 1$) or not ($x_i = 0$). The decision $x_i = 1$ means that player $i$ buys the data, and hence trade $i$ occurs, provided that all previous trades have occurred. We therefore use $y_i \coloneqq \prod_{j\le i} x_j$ to indicate whether trade $i$ actually occurs and player $i$ obtains the data. If $y_i = 1$, trade $i$ occurs, player $i$ realizes valuation $v_i$, and pays $p_{i-1}$; otherwise, $y_i=0$ and player $i$ receives no valuation and makes no payment. At each stage $i$, the buyer observes their own valuation $v_i$ and the public transaction history
\begin{equation}
\label{eq:hi}
h_i = (p_0,x_1,p_1,\ldots,p_{i-2},x_{i-1},p_{i-1}).
\end{equation}
When $i=0$, we let $h_0 = \perp$. We denote $H_i$ as all possible values of $h_i$.

\begin{remark}
Real-world trading may involve richer trading dynamics than posted-price dynamic, but posted-price dynamic is common in platform-mediated markets (\eg, financial exchanges, e-commerce platforms): sellers announce a price, and buyers decide whether to accept the offer. Posted-price dynamic is also a standard modeling choice in the sequential trading literature \citep{condorelli2017bilateral,manea2018intermediation}.
\end{remark}

A strategy for player $i$ is a function pair $(\tx_i, \tp_i)$\footnote{We use $\tilde{\cdot}$ (\eg, $\tx$) to denote a strategy and the corresponding non-tilde variable (\eg, $x$) to denote the action induced by that strategy.}, where $\tx_i: \calV_i \times H_i \to \{0, 1\}$ and $\tp_i: \calV_i \times H_i \to \mathbb{R}_+$ are the buying and pricing strategies, respectively. We call $(\tbmx,\tbmp)$ a strategy profile, where $\tbmx=\{\tx_i\}_\iinn$ and $\tbmp=\{\tp_i\}_{0\le i<n}$.

\begin{remark}
Player $0$ has no buying strategy, and player $n$ has no pricing strategy.
\end{remark}

\begin{definition}[Game Dynamics]
\label{def:game-evolve}
Given a valuation profile $\bmv$ and a strategy profile $(\tbmx, \tbmp)$, the game evolves as follows: \\
1. For each round $i \in \{1, \dots, n\}$, player $i-1$ posts a price $p_{i-1} = \tp_{i-1}(v_{i-1}, h_{i-1})$ to player $i$. \\
2. Player $i$ then responds with the willingness to buy $x_i = \tx_i(v_i, h_i)$. \\
3. The game concludes with the realization of the full history $h = (p_0, x_1, \dots, p_{n-1}, x_n)$.
\end{definition}

\begin{figure*}[t]
\centering
\includegraphics[width=\textwidth]{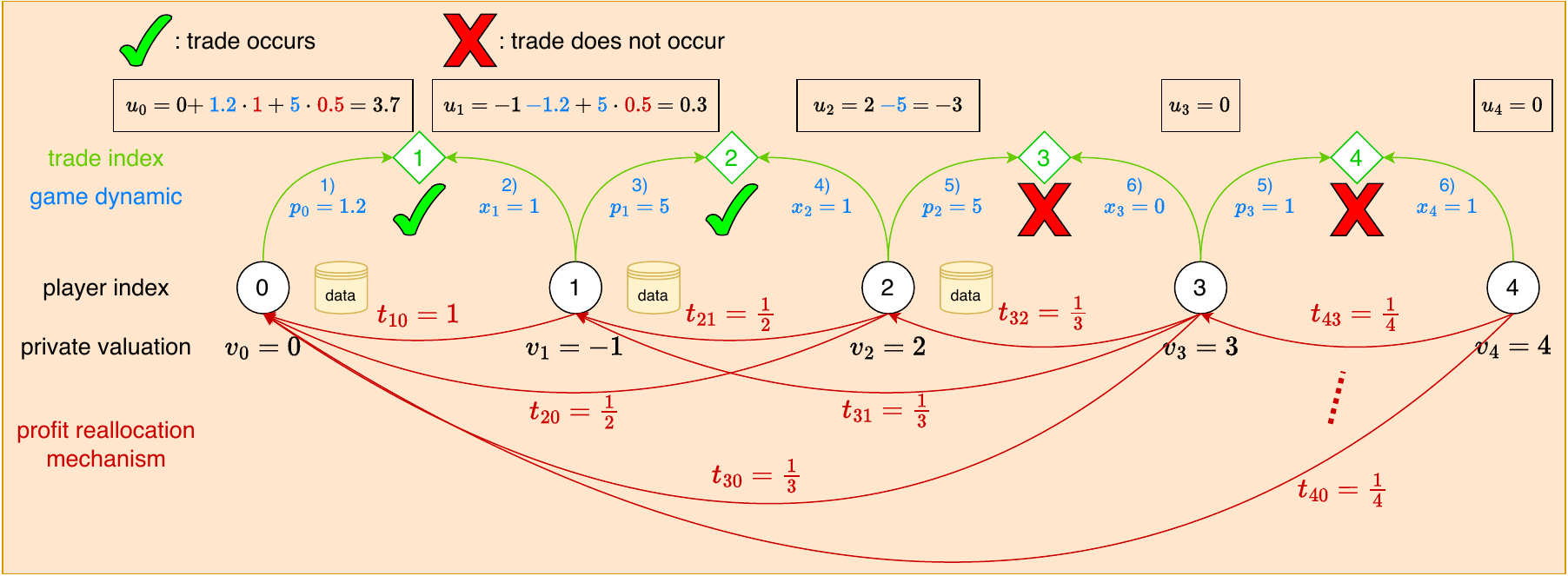}
\caption{A running example with $n=4$. Best viewed in colors (same for remaining figures).}
\label{fig:example}
\end{figure*}

\subsection{Profit Reallocation Mechanism}\label{subsec:mechanism}


A \emph{profit reallocation mechanism} is formally defined by a nonnegative matrix $\bmt\in\bbR_+^{n\times (n+1)}$, which specifies how profits from trades are shared among players.
Specifically, $t_{ij}$ is the proportion of profit from \emph{trade} $i$ allocated to player $j$: if trade $i$ occurs at price $p_{i-1}$, the mechanism first charges player $i$ the payment $p_{i-1}$ and then reallocates this profit among players $0$ to $i-1$. Player $j$ receives $p_{i-1}\cdot t_{ij}$ from the mechanism.
The model is therefore specified by $\mathrm{Model}=(n,\bmt,\calF)$.

To ensure the proposed mechanism is well-defined and economically viable, we impose the following two natural conditions:

\begin{definition}[Budget Feasibility]
\label{def:mechanism:1}
The total proportion of profit redistributed from any trade $i$ must not exceed the realized profit. Formally, for all $i \in [n]$, $\sum_{j=0}^n t_{ij} \leq 1$.
\end{definition}

\begin{definition}[Past Only]
\label{def:mechanism:2}
Only players who participate in trades \emph{prior to} trade $i$ are eligible to receive a portion of the reallocated profit from trade $i$. This implies $t_{ij} = 0$ for all $j \geq i$.
\end{definition}

The baseline mechanism, in which the seller (player $i-1$) retains the entirety of the profit in trade $i$, constitutes a special case where $t_{i,i-1}=1$ and $t_{ij}=0$ for all $j \neq i-1$.

\paragraph{Comparison to TBDS \citep{xin2025tbds}.}
Transaction-Based Data Sharing (TBDS) reallocates profits along a chain using a single parameter $\alpha\in[0,1]$: 
when player $i-1$ earns profit from player $i$, an $\alpha$ fraction is redistributed to its upstream seller (player $i-2$, if any). Therefore, TBDS is a special case of our mechanism where $t_{ij} = (1-\alpha)\cdot \alpha^{i-j-1}$ if $0<j<i$ and $t_{i0} = \alpha^{i-1}$.

\subsection{Utility and Solution Concept}\label{subsec:solution}

Given a profit reallocation mechanism $\bmt$ and the full history $h$, 
the deterministic utility of player $\iinn$ with private valuation $v_i$ is
\begin{equation}\label{eq:ui-history}
U_i(h; v_i) = (v_i-p_{i-1})\,\prod_{k=1}^i x_k
+ \sum_{j=i+1}^n (t_{ji} p_{j-1}\prod_{k=1}^j x_k).
\end{equation}
Recall that $\prod_{k=1}^j x_k = y_j \in \{0,1\}$ is an indicator that trade $j$ occurred, and $h = (p_0, x_1, \dots, p_{n-1}, x_n)$ is the full game history.
When $i=0$, we simply let $p_{-1} = 0$.
In \cref{eq:ui-history}, the first term is the surplus of player $i$ from trade $i$.
The second term represents the profit reallocated to player $i$ from \emph{subsequent trades} $j > i$.
(By \cref{def:mechanism:2}, player $i$ does not benefit from previous trades $j<i$.)
Specifically, when trade $j$ occurs ($\prod_{k=1}^j x_k$) at price $p_{j-1}$, player $i$ receives a reallocated profit of $t_{ji} p_{j-1}$.
We define social welfare as the sum of players' utilities, \ie, $\SW(h;\bmv) = \sum_{i=0}^n u_i(h;v_i)$.

\begin{example}
We provide an example in \cref{fig:example} with $n=4$ where players $0,1,2,3,4$ engage in trades $1,2,3,4$.
The valuation profile, the profit reallocation mechanism, and all players' actions are illustrated in the figure; together, they uniquely determine all players' utilities.
Although player $4$ chooses $x_4=1$ (willingness to buy), trade $4$ does not occur because an upstream trade (trade $3$) fails. Consequently, player $4$ obtains zero utility.
\end{example}

A critical issue is which solution concept best predicts market outcomes.
Because trades occur on the platform, players typically do not know one another, making collusion unlikely.
Therefore, 
we adopt \emph{sequential equilibria}~\cite{kreps1982sequential} as a natural solution concept for non-cooperative games with sequential, incomplete information and individual utility maximization. 
This concept requires each player to maximize her expected utility given the information available to her.
We leave the study of other solution concepts (\ie, anti-collusion and sybil-proofness) to future work.

Specifically, let the players' belief profile be $\tilde{\bm\mu} = \{\tmu_i\}_{0\le i\le n}$, where $\tilde{\mu}_i(\cdot|v_i, h_i) \in \Delta(\calV)$ represents player $i$'s belief about the valuation profile $\bmv$ based on her current information, including her private valuation $v_i$ and the observed history $h_i$.

\begin{definition}[Sequential equilibria~\cite{kreps1982sequential}]
\label{def:equilibrium}
A sequential equilibrium consists of a strategy profile and a belief profile
$(\tilde{\bmx},\tilde{\bmp},\tilde{\bmmu})$ satisfying the following two conditions:
\begin{itemize}[left=0em]
\item \textbf{Rationality}:
For every player $i$ and every information set $(v_i, h_i)$, the action pair $(x_i = \tx_i(v_i,h_i), p_i = \tp_i(v_i,h_i))$ maximizes her expected utility under the belief $\tmu_i(\cdot | v_i, h_i)$:
\begin{align}
\label{eq:sequential-eq}
\bbE_{\bmv\sim \tmu_i(\cdot|v_i, h_i)}[U_i(h;v_i)] 
\geq \bbE_{\bmv\sim \tmu_i(\cdot|v_i, h_i)}[U_i(h';v_i)],
\end{align}
for all deviating actions $(x'_i, p'_i)$, where the full history $h$ (\resp, $h'$) is generated according to \cref{def:game-evolve} under valuation profile $\bmv$ when player $i$ takes action $(x_i, p_i)$ (\resp, $(x'_i, p'_i)$) and all other players follow $(\tbmx_{-i},\tbmp_{-i})$.
In other words, each player maximizes expected utility based on her posterior beliefs, as is standard in the economics literature \citep{MWG}.

\item \textbf{Posterior Belief:} For every player $i$, the belief $\tmu_i(\cdot|v_i,h_i)$ is the posterior of $\bmv$ induced by the prior $\calF$, the observed information $(v_i,h_i)$, and the strategy profile $(\tbmx,\tbmp)$.
\end{itemize}
\end{definition}

Under the strategy and belief profiles $(\tilde{\bmx},\tilde{\bmp},\tilde{\bm\mu})$, player $i$'s expected utility from choosing action pair $(x_i,p_i)$ given information $(v_i, h_i)$ is
\begin{equation}
\label{eq:utility}
u_i(x_i, p_i;v_i,h_i) = \big(\prod_{j\le i} x_j\big)
\Big[(v_i-p_{i-1})+\bbE_{\bmv_{i+1:n}\sim \tmu_i(\cdot|v_i,h_i)}
\big(\sum_{i < j \le n} t_{ji} p_{j-1}\prod_{i< k \le j} x_{k}\big)\Big],
\end{equation}
where $x_k=\tilde x_k(v_k,h_k)$, $p_k=\tilde p_k(v_k,h_k)$, and $h_k$ is recursively defined by \cref{eq:hi} for $k>i$.
Without loss of generality, we view $\tmu_i$ as a belief over $\bmv_{i+1:n}$, since player $i$'s expected utility in \cref{eq:utility} does not depend on $\bmv_{0:i-1}$.

\section{Equilibrium Computation}
\label{sec:algo}

This section presents algorithms for computing a sequential equilibrium given $n$, a mechanism $\bmt$, and a valuation distribution $\calF$.
In the original game, each player $i$'s strategy is a high-dimensional function, which is difficult to analyze and compute.
We first establish a crucial simplification (\cref{subsec:simplification}) for sequential equilibrium under the Markovian assumption (\cref{def:markov}).
Building on this simplification, we provide a polynomial-time exact algorithm when valuations have finite support (\cref{subsec:exact-finite}).
We then provide an FPTAS for the more general case of continuous valuation support (\cref{subsec:approx-continuous}).

\subsection{Equilibrium Simplification}
\label{subsec:simplification}


We first simplify equilibrium strategies to single-variable functions.
In a sequential equilibrium, player $i$'s posterior belief depends only on $v_i$, regardless of the history $h_i$ and strategy profile $(\tbmx,\tbmp)$.

\begin{fact}
\label{fact:simplify}
Let $\calF$ satisfy \cref{def:markov} and $(\tbmx,\tbmp,\tbmmu)$ be a sequential equilibrium.
We have $\tmu_i(\bmv_{i+1:n}|v_i,h_i) = \calF(\bmv_{i+1:n}|v_i)$.
\end{fact}

\cref{fact:simplify} follows from \cref{def:markov} and the observation that, for any fixed strategy profile and fixed $v_i$, $h_i$ is independent of $\bmv_{i+1:n}$.
We next present the core simplifications, established by \cref{lem:simplify-sell} and \cref{lem:simplify-buy}.

\begin{restatable}{lemma}{lemSimplifySell}
\label{lem:simplify-sell}
Let $\calF$ satisfy \cref{def:markov}.
There exists a sequential equilibrium $(\tbmx,\tbmp,\tbmmu)$ such that the following holds for some functions $\{\tx^\dagger_i(v,p)\}_\iinn$ and $\{\tp^\dagger_i(v)\}_{0\le i < n}$:
\begin{equation}
\label{eq:lem:simplify-sell}
\begin{aligned}
\tx_i(v_i, h_i) = \tx^\dagger_i(v_i, p_{i-1}) ,\quad \forall v_i, h_i, \forall 1\le i\le n\quad
\tp_i(v_i, h_i) = \tp^\dagger_i(v_i) ,\quad \forall v_i, h_i, \forall 0\le i < n
\end{aligned}
\end{equation}

\end{restatable}

Due to space constraints, we defer the full proof of \cref{lem:simplify-sell}, together with the proofs of other statements, to \cref{sec:omitted-proofs}.


\begin{restatable}{corollary}{lemSimplifyBuy}
\label{lem:simplify-buy}
Let $\calF$ satisfy \cref{def:markov}.
There is a sequential equilibrium $(\tbmx,\tbmp,\tbmmu)$ such that the conclusions of \cref{lem:simplify-sell} hold and further, for each player $\iinn$, $\tx^\dagger_i(v_i,p_{i-1})$ is a threshold strategy:
\begin{equation}
\label{eq:lem:simplify-buy:1}
\begin{aligned}
\tilde x^\dagger_i(v_i, p_{i-1})=
\begin{cases}
1, & \text{if } \ts^\dagger_i(v_i)\ge p_{i-1}\\
0, & \text{otherwise.}
\end{cases}
\end{aligned}
\end{equation}
for some function $\ts^\dagger_i(v_i)$.
Specifically, $\ts^\dagger_i(v_i)$ has the form:
\begin{equation}
\label{eq:lem:simplify-buy:2}
\begin{aligned}
\ts^\dagger_i(v_i) = v_i + \max_{p_i \in \bbR_+} \bbE_{\bmv_{i+1:n}\sim \calF(\cdot|v_i)}
\Big[\sum_{i<j\leq n} \Big( t_{ji} p_{j-1} \prod_{i<k\leq j}x_k \Big)\Big],
\end{aligned}
\end{equation}
where the maximum value is achieved by $\tp^\dagger_i(v_i)$. $\{p_j,x_j\}_{j > i}$ in \cref{eq:lem:simplify-buy:2} is determined by $\bmv_{i+1:n}$ and game dynamic as in \cref{def:game-evolve}.
\end{restatable}


$\ts^\dagger_i(v_i)$ in \cref{lem:simplify-buy}
represents player $i$'s expected gain from buying the data from player $i-1$ and setting the optimal price $p_i$ to player $i+1$.
\cref{lem:simplify-buy,lem:simplify-sell} allow us to restrict attention to sequential equilibria with simplified, single-variable expressions $(\tp^\dagger_i(v_i),\ts^\dagger_i(v_i))$ in \cref{eq:lem:simplify-sell,eq:lem:simplify-buy:1}.
Therefore, with a slight abuse of notation, we use $\tp_i(v_i)$ and $\ts_i(v_i)$ to denote player $i$'s pricing and threshold strategies, $(\tbmp,\tbms) = (\{\tp_i\}_{0\le i < n}, \{\ts_i\}_\iinn)$ to denote a strategy profile, and $(\tbmp^*,\tbms^*)$ to denote an equilibrium strategy profile. The original strategy profile $(\tbmx,\tbmp,\tbmmu)$ can be recovered from $(\tbmp,\tbms)$ using \cref{eq:lem:simplify-sell,eq:lem:simplify-buy:1,fact:simplify}.

Given a strategy profile $(\tbmp,\tbms)$, we define an auxiliary gain function for player $i$ from buying the data and setting price $p_i$ to player $i+1$:
\begin{equation}
\label{eq:lem:simplify-buy:3}
g_i(v_i,p_i) = v_i + \bbE_{\bmv_{i+1:n}\sim \calF(\cdot|v_i)}\Big[\sum_{i<j\leq n}\Big( t_{ji} p_{j-1} \prod_{i<k\leq j}x_k\Big)\Big].
\end{equation}
We denote $g_i(v_i,p_i)$ as player $i$'s gain function.
By \cref{lem:simplify-buy}, a sufficient condition for $(\tbmp,\tbms)$ to form a sequential equilibrium is
\begin{equation}
\label{SE:transform}
\tp_i(v_i) \in \argmax_{p_i \in \bbR_+} g_i(v_i,p_i), \quad \ts_i(v_i) = g_i(v_i,\tp_i(v_i)), \quad 0 \le i \le n
\end{equation}
In the sequel, we focus on computing sequential equilibria that satisfy \cref{SE:transform}.



\subsection{Polynomial Algorithm for Finite Valuation Support}
\label{subsec:exact-finite}

We first consider the case where the model has finite valuation support and provide a polynomial-time algorithm for computing an exact sequential equilibrium.


\begin{assumption}[Finite Valuation Support]\label{asp:finite}
We say that $\calF$ has finite valuation support if $\Pr[\forall i, v_i\in W]=1$ for $\bmv\sim\calF$, where $W=\{w_1,\ldots,w_L\}\subseteq \bbR$.
\end{assumption}

Under \cref{asp:finite}, player $i$'s strategy can be represented by finitely many entries: $p^*(i,k) \coloneqq \tp^*_i(w_k)$ and $u^*(i,k) \coloneqq \ts^*_i(w_k)$ for $\iinn, \kinL$.
The conditional distribution $f_i(v_i|v_{i-1})$ also has finite representation:
$Q = \{q_{i,k,\ell}\}_{i\in[n], k,\ell\in[L]}$, where $q_{i,k,\ell}\coloneqq f_i(v_i = w_\ell|v_{i-1} = w_k)$ is the conditional density of $v_i$ given $v_{i-1} = w_k$, satisfying $\sum_{\ell\in[L]}q_{i,k,\ell}=1$ for each $(i,k)$.

\begin{restatable}{theorem}{thmAlgoExact}
\label{thm:algo-exact}
Let \cref{def:markov,asp:finite} hold.
Then, given the model representation $(n,L,W,\bmt, Q)$,
\Cref{alg:exact} (deferred to \cref{sec:omitted-algo}) outputs an exact sequential equilibrium in $O(n^2L^4)$ running time.
\end{restatable}

\cref{thm:algo-exact} enables efficient evaluation of a given profit reallocation mechanism, and therefore supports the experimental comparison of different mechanisms.

\subsection{FPTAS for Continuous Valuation Support}
\label{subsec:approx-continuous}

For games with continuous valuation support, finding an exact sequential equilibrium is technically challenging because the game generally has no finite representation. 
We therefore seek an $\varepsilon$-approximate sequential equilibrium.


\begin{definition}[$\varepsilon$-approximate sequential equilibrium]\label{def:eps-equilibrium}
A strategy profile
$(\tbmp,\tbms)$ is an $\varepsilon$-approximate sequential equilibrium if, for every player $i$, it satisfies: (a) $g_i(v_i,\tp_i(v_i)) - \ts^\dagger_i(v_i) \ge - \varepsilon \cdot \max\{1, \ts^\dagger_i(v_i)\}$; and (b) $|\ts_i(v_i) - \ts^\dagger_i(v_i)| \le \varepsilon \cdot \max\{1, \ts^\dagger_i(v_i)\}$,
where $\ts^\dagger_i(v_i)$ and $g_i(v_i,p_i)$ are defined by \cref{eq:lem:simplify-buy:2,eq:lem:simplify-buy:3}, respectively.
\end{definition}

Note that \cref{def:eps-equilibrium} reduces to the definition of exact sequential equilibrium when $\varepsilon=0$, making \cref{def:eps-equilibrium} a meaningful measure of approximation.

We show that an FPTAS exists for computing an $\varepsilon$-approximate sequential equilibrium.
The result requires a smoothness condition on the conditional distribution $f_i(v_i|v_{i-1})$.


\begin{definition}[Log-Lipschitzness]\label{asp:lipschitz}
Under \cref{def:markov}, we say that $f_i$ is $K_0$-log-Lipschitz if, for each $\iinn$, $f_i(v_i | v_{i-1})$ is $K_0$-log-Lipschitz in $v_{i-1}$: for all $v_i>0$ and $a,b>0$,
$\bigl|\log f_i(v_i | a) - \log f_i(v_i | b)\bigr|\le K_0 \,| a - b|$.
\end{definition}


\begin{restatable}{theorem}{thmAlgoApprox}
\label{thm:algo-approx}
Under \cref{def:markov}, assume that $f_i$ is $K_0$-log-Lipschitz for all $\iinn$, and that each $v_i$ has support $[L_0,H_0]\subseteq \bbR$. 
Then, given the model representation $(n,\calF, \{t_{ij}\}_{0\le i,j\le n})$, \cref{alg:approx} (deferred to \cref{sec:omitted-algo}) outputs an $\varepsilon$-approximate sequential equilibrium in 
$\poly\bigl(n,\varepsilon^{-1},H_0 - L_0, K_0\bigr)$ time.
\end{restatable}

The high-level idea of \cref{alg:approx} is to discretize the continuous valuation supports and compute exact equilibria in the discretized model. We prove \cref{thm:algo-approx} by showing that exact equilibria of the discretized model are approximate equilibria of the original model. 
\cref{thm:algo-approx} indicates that discretized models provide a tractable and provably accurate approximation for analyzing continuous models.


\section{Equilibrium Comparison between Profit Reallocation Mechanisms}
\label{sec:theoretic}

This section establishes key results about the equilibrium comparison between profit reallocation mechanisms.
\cref{thm:global} is our main finding: compared with the baseline mechanism ($t_{ij}=\one\{i=j+1\}$), profit reallocation mechanism expands the equilibrium trade event.

There may be multiple sequential equilibria under a mechanism $\bmt$, and some may not satisfy the simplification established in \cref{subsec:simplification}.
To make the equilibrium under $\bmt$ well-defined, we consider the simplified equilibrium ($\tbmp^*,\tbms^*$) established by \cref{lem:simplify-sell,lem:simplify-buy}, and require the tie-breaking rule to be independent of the mechanism $\bmt$.
This convention is primarily technical, but it is also behaviorally natural: players use the single-variable strategies characterized in \cref{subsec:simplification} rather than more complex history-dependent strategies, and their tie-breaking rules do not depend on the mechanism.
To ensure tractability, we also impose two assumptions:

\begin{assumption}[Positivity]
\label{asp:posi}
Positivity requires $\Pr_{\bmv \sim \calF}[v_i > 0] = 1$ for all $0 \le i\le n$.
\end{assumption}

\begin{assumption}[Homogeneity]
\label{asp:homo}
Let \cref{def:markov} hold.
Homogeneity requires that $\forall\alpha,w,v>0$ and $\iinn$, we have $F_i(\alpha v|\alpha w)=F_i(v|w)$.
\end{assumption}

The Positivity assumption states that the value generated from data always exceeds the processing cost.
The Homogeneity assumption is widely used in economic models, especially market models \citep[\S6.2]{AGT:nisan_algorithmic_2007}, \citep[\S3.2.3]{DSGE}.
It captures settings in which production output scales linearly with the input.

\subsection{Local Comparison between Mechanisms}
\label{subsec:partial}

We first compare the equilibria induced by two profit reallocation mechanisms $\bmt^1$ and $\bmt^2$ that differ locally. This result provides a building block for comparing any profit reallocation mechanism with the baseline mechanism by connecting them through finitely many locally different mechanism pairs.

\cref{thm:partial} is the main theorem of \cref{subsec:partial}. It shows that increasing future reallocated profit and decreasing one-shot profit for player $i$ expands the equilibrium trade event of trade $i+1$ while leaving other equilibrium trade events unchanged.

\begin{restatable}{theorem}{thmPartial}\label{thm:partial}
Let the valuation distribution satisfy \cref{def:markov,asp:posi,asp:homo}.
Fix player index $i$. Let $\bmt^1,\bmt^2\in\bbR_+^{n\times (n+1)}$ be two profit reallocation mechanisms that differ only in how revenue is allocated to player $i$.
Specifically: (a) $t^1_{j,i'}=t^2_{j,i'}$ for all $j>i'>i$; (b) $t^1_{j,i}\le t^2_{j,i}$ for all $j\ge i+2$; (c) $t^1_{i+1,i}\ge t^2_{i+1,i}$; (d) $t^k_{j,i'}=0$ for all $i'<i$, $k=1,2$ and $\jinn$ unless $j=i'+1$. 
Let $(\tbmp^{k,*},\tbms^{k,*})$ be the corresponding equilibria for $\bmt^k$, $k\in\{1,2\}$, and define $E_i^k$ as the set of valuation profiles under which player $i$ is willing to buy: $E_i^k = \{\bmv\in\calV: \one\{\ts^{k,*}_i(v_i) \ge \tp^{k,*}_{i-1}(v_{i-1})\}\}$.
Then: 1. $E_j^1 = E_j^2$ for $j \ne i+1$; 2. $E_{i+1}^1 \subseteq E_{i+1}^2$.
\end{restatable}

\subsection{Comparison with Baseline Mechanism}
\label{subsec:global}

We now leverage \cref{thm:partial} to compare the baseline mechanism $\bmt$ (where $t_{ij}=\one\{i=j+1\}$) against an arbitrary feasible profit reallocation mechanism.
We prove that profit reallocation mechanism weakly expands the equilibrium trade event of each trade $i$.

\begin{restatable}{theorem}{thmGlobal}
\label{thm:global}
Let \cref{def:markov,asp:posi,asp:homo} hold.
Let $\bmt^1$ denote the baseline mechanism $t^1_{ij}=\one\{i=j+1\}$, and $\bmt^2$ be any other feasible profit reallocation mechanism that satisfies \cref{def:mechanism:1,def:mechanism:2}.
Denote $(\tbmp^{1,*},\tbms^{1,*})$ and $(\tbmp^{2,*},\tbms^{2,*})$ as the corresponding equilibria. 
Then, $E^1_i \subseteq E^2_i$ for all $\iinn$.
\end{restatable}

Note that $\bmv\in \cap_{\jlei} E^k_j$ ($k=1,2$) means that trade $i$ occurs under mechanism $\bmt^k$, valuation profile $\bmv$, and the players' equilibrium strategies.
Under the assumptions of \cref{thm:global}, if we further require the mechanism be budget-balanced, \ie, $\sum_{j} t_{ij} = 1$, then, the occurrence of trade $i$ increases social welfare by $v_i \ge 0$.
Overall, $E^1_i \subseteq E^2_i$ implies that, compared with the baseline mechanism, any feasible profit reallocation mechanism weakly expands trade events, and any feasible, budget-balanced profit reallocation mechanism weakly improves social welfare under equilibrium.

\section{Experimental Evaluation}
\label{sec:exp}

This section presents computational experiments that complement the theoretical analysis and evaluate the empirical effects of profit reallocation.
The experiments are designed to address the following questions that are not fully resolved by the theory:
\begin{enumerate}
\item[(a)] To what extent do profit reallocation mechanisms improve upon the baseline mechanism?
\item[(b)] Do these improvements persist when \cref{asp:homo,asp:posi} are not satisfied?
\item[(c)] Are the improvements robust across valuation distributions $\calF$ and chain lengths $n$?
\item[(d)] Does profit reallocation adversely affect the utilities of individual players?
\end{enumerate}

\subsection{Experimental Setup}

We construct four finite-support valuation distributions $\calF$ to evaluate the robustness of our theoretical predictions.
All models satisfy \cref{def:markov}, which enables exact equilibrium computation via \cref{alg:exact}. Recall that a game with chain length $n$ consists of players $0$--$n$ and trades $1$--$n$.


\paragraph{Valuation models.}
All four models use a five-neighbor Markov transition: given $v_{i-1}=w_k$, draw $k'\sim\mathrm{Uniform}(\{k-2,k-1,k,k+1,k+2\})$ and set $v_i=w_{k'}$, except that Model 2 truncates $k'$ to $[0,2n]$.
They differ only in the support and valuation grid:
\begin{itemize}[left=0em]\itemsep0.1em
\item \textbf{Model 1} uses $L=4n+1$, $W=\{w_{-2n},\ldots,w_{2n}\}$, $w_k=\exp(k/2)$, and $v_0=w_0$; it satisfies both \cref{asp:homo,asp:posi}.
\item \textbf{Model 2} uses $L=2n+1$, $W=\{w_0,\ldots,w_{2n}\}$, $w_k=k/2n$, and $v_0=w_0$; it violates homogeneity but preserves positivity.
\item \textbf{Model 3} uses $L=4n+1$, $W=\{w_{-2n},\ldots,w_{2n}\}$, $w_k=(-e)^{k/2}$, and $v_0=w_0$; it violates positivity but preserves homogeneity.
\item \textbf{Model 4} uses $L=4n+1$, $W=\{w_{-2n},\ldots,w_{2n}\}$, $w_k=k/2n$, and $v_0=w_0$; it violates both homogeneity and positivity.
\end{itemize}
Model 1 is covered by our theory, while Models 2--4 test robustness when one or both assumptions fail.


\paragraph{Profit reallocation mechanisms.}
We compare three profit reallocation mechanisms:
\begin{itemize}[left=0em]
\item RAW (baseline mechanism): the seller retains the entire profit ($t_{i,i-1}=1$, $t_{ij}=0$ when $j \neq i-1$).
\item AVE: profit is reallocated uniformly to upstream players ($t_{ij} = 1/i$ for all $j \in \{0,\ldots,i-1\}$).
\item OPT: The mechanism that maximizes equilibrium social welfare subject to budget balance.
\end{itemize}

For RAW and AVE, the profit reallocation mechanism $\bmt$ is fixed in advance; we then run \cref{alg:exact} on each model to obtain the exact equilibrium induced by that mechanism.
For OPT, we parameterize $\bmt$ and, for each candidate mechanism, compute the induced equilibrium $\mathrm{eq}(\bmt,M)$ using \cref{alg:exact} on model $M$. We then search for $\bmt^*(M)$ that maximizes equilibrium social welfare $\SW(\mathrm{eq}(\bmt,M),\bmt,M)$.
We compute $\bmt^*(M)$ using the zeroth-order optimization procedures detailed in Appendix~\ref{sec:omitted_exp}.

We evaluate two metrics. The first is expected social welfare, $\SW = \bbE_{\bmv\sim\calF}[\sum_{i=0}^n u_i(h;v_i)]$, defined as the sum of player-wise utilities. The second is transaction volume, $\mathrm{TV} = \bbE_{\bmv\sim\calF}[\sum_{i=1}^n \one\{\bmv \in \cap_{1\le j\le i} E_j\}]$, defined as the expected number of successful trades, where $\one\{\bmv \in \cap_{1\le j\le i} E_j\}$ indicates that trade $i$ occurs and $E_j$ is defined in \cref{thm:global}.
Both metrics are evaluated under equilibrium strategies. For each seed, we generate $M=10,000$ independent samples of $\bmv$ from $\calF$ to estimate these expectations.

\subsection{Results and Analysis}

\begin{table*}[t]
\centering
\caption{Social welfare and transaction volume across mechanisms and models with chain length $n=5$. We report mean value and standard deviations with $5$ different seeds.}
\adjustbox{max width=\textwidth}{
\begin{tabular}{l|cccc|cccc}
\toprule
Metrics: & \multicolumn{4}{c}{Social Welfare} & \multicolumn{4}{c}{Transaction Volume} \\
\cmidrule(lr){1-1}\cmidrule(lr){2-5}\cmidrule(lr){6-9}
Mechanisms & Model 1 & Model 2 & Model 3 & Model 4 & Model 1 & Model 2 & Model 3 & Model 4 \\
\midrule
RAW &
4.400$_{\pm 0.053}$ & 0.232$_{\pm 0.003}$ & 2.021$_{\pm 0.008}$ & 0.241$_{\pm 0.002}$ &
0.664$_{\pm 0.011}$ & 0.858$_{\pm 0.011}$ & 0.660$_{\pm 0.013}$ & 0.876$_{\pm 0.012}$ \\
AVE &
6.973$_{\pm 0.073}$ & 0.474$_{\pm 0.001}$ & 2.575$_{\pm 0.012}$ & 0.269$_{\pm 0.002}$ &
1.363$_{\pm 0.013}$ & 3.055$_{\pm 0.003}$ & 1.353$_{\pm 0.014}$ & 1.001$_{\pm 0.014}$ \\
OPT &
10.804$_{\pm 0.099}$ & 0.621$_{\pm 0.002}$ & 2.577$_{\pm 0.012}$ & 0.337$_{\pm 0.004}$ &
4.401$_{\pm 0.002}$ & 4.524$_{\pm 0.002}$ & 1.353$_{\pm 0.014}$ & 1.519$_{\pm 0.028}$ \\
\bottomrule
\end{tabular}
}
\label{tab:exp}
\end{table*}

\begin{figure*}[t]
\centering
\begin{subfigure}{0.235\linewidth}
\includegraphics[width=\linewidth]{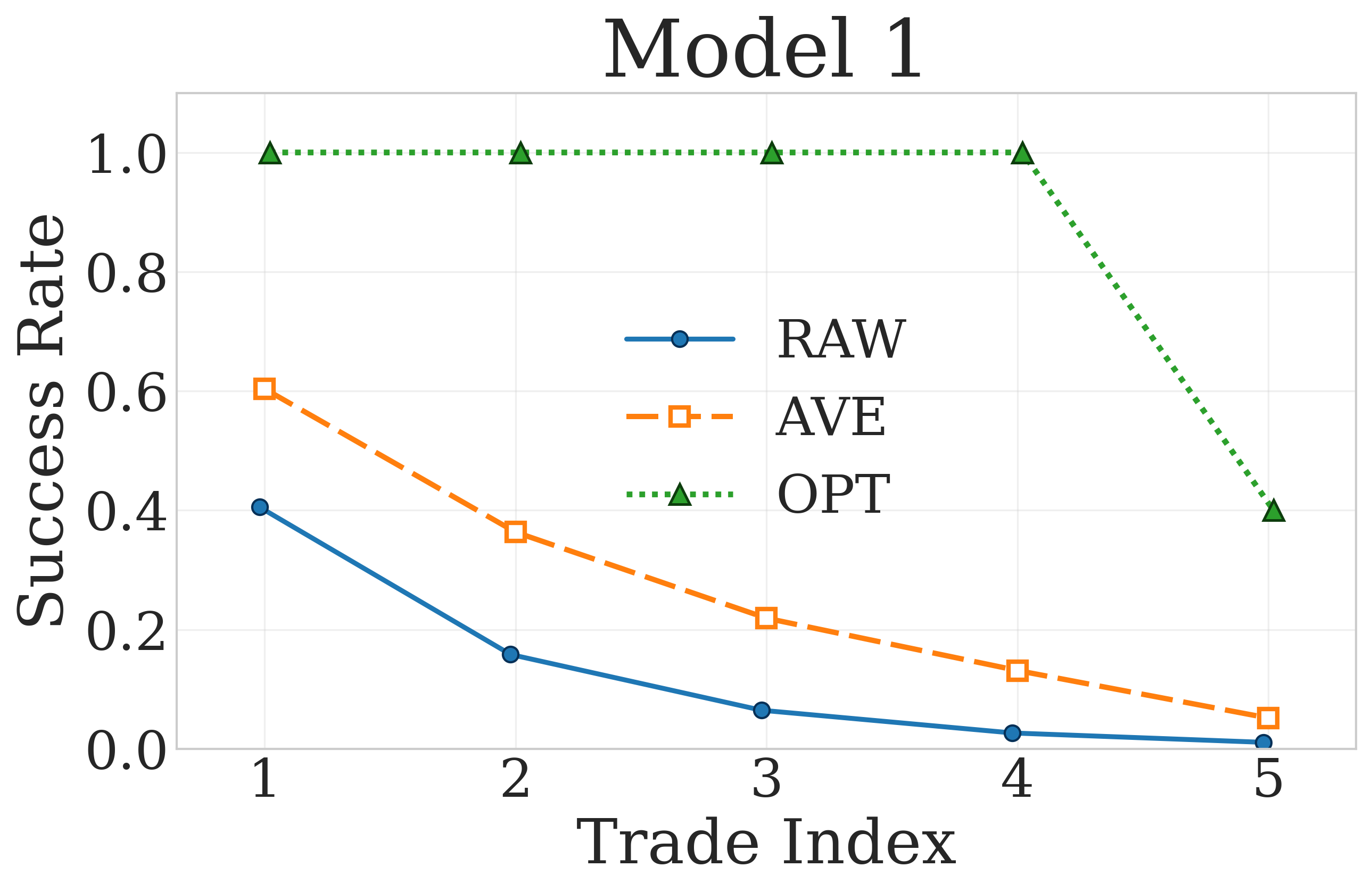}
\end{subfigure}
\hfill
\begin{subfigure}{0.235\linewidth}
\includegraphics[width=\linewidth]{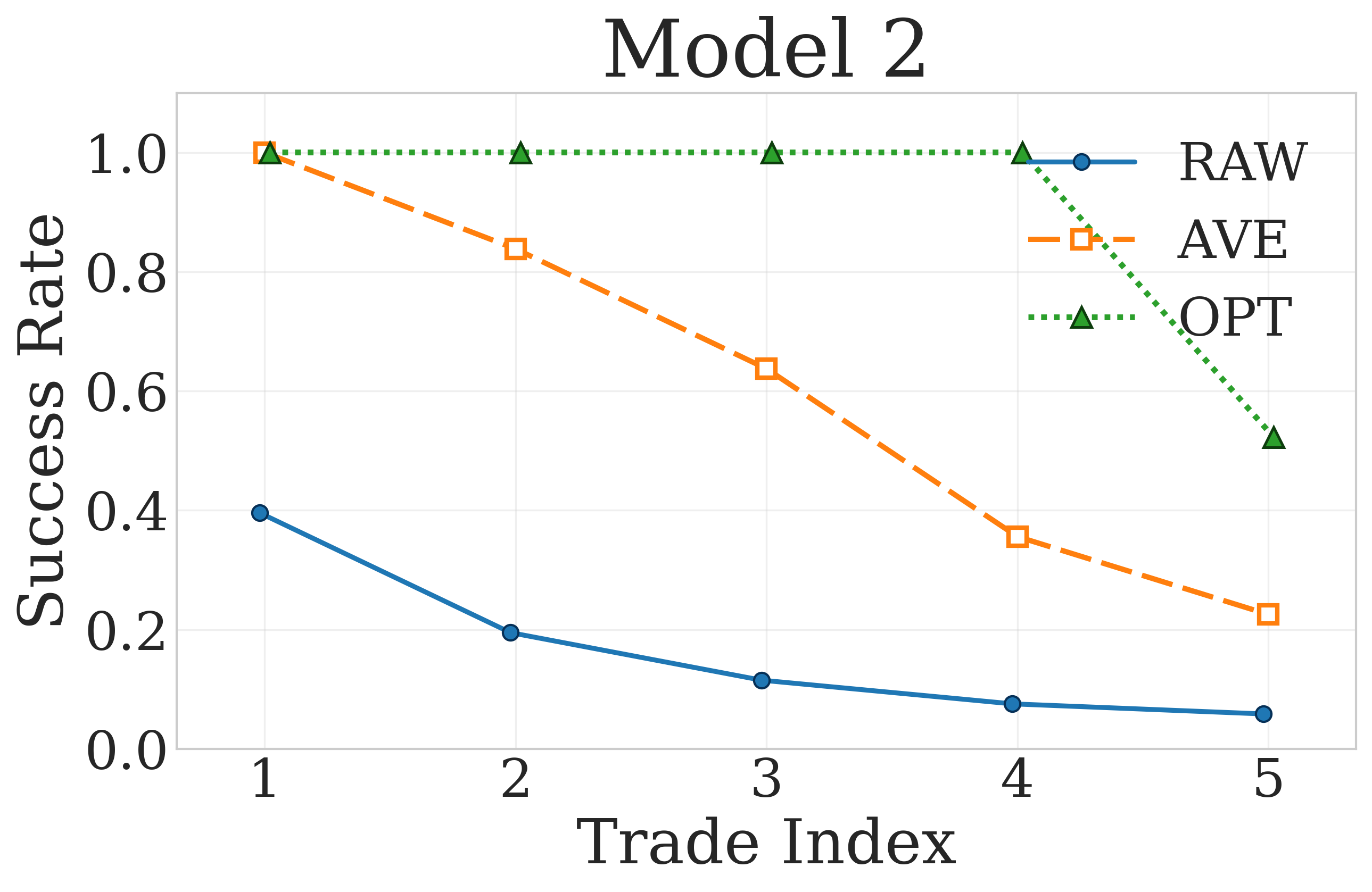}
\end{subfigure}
\hfill
\begin{subfigure}{0.235\linewidth}
\includegraphics[width=\linewidth]{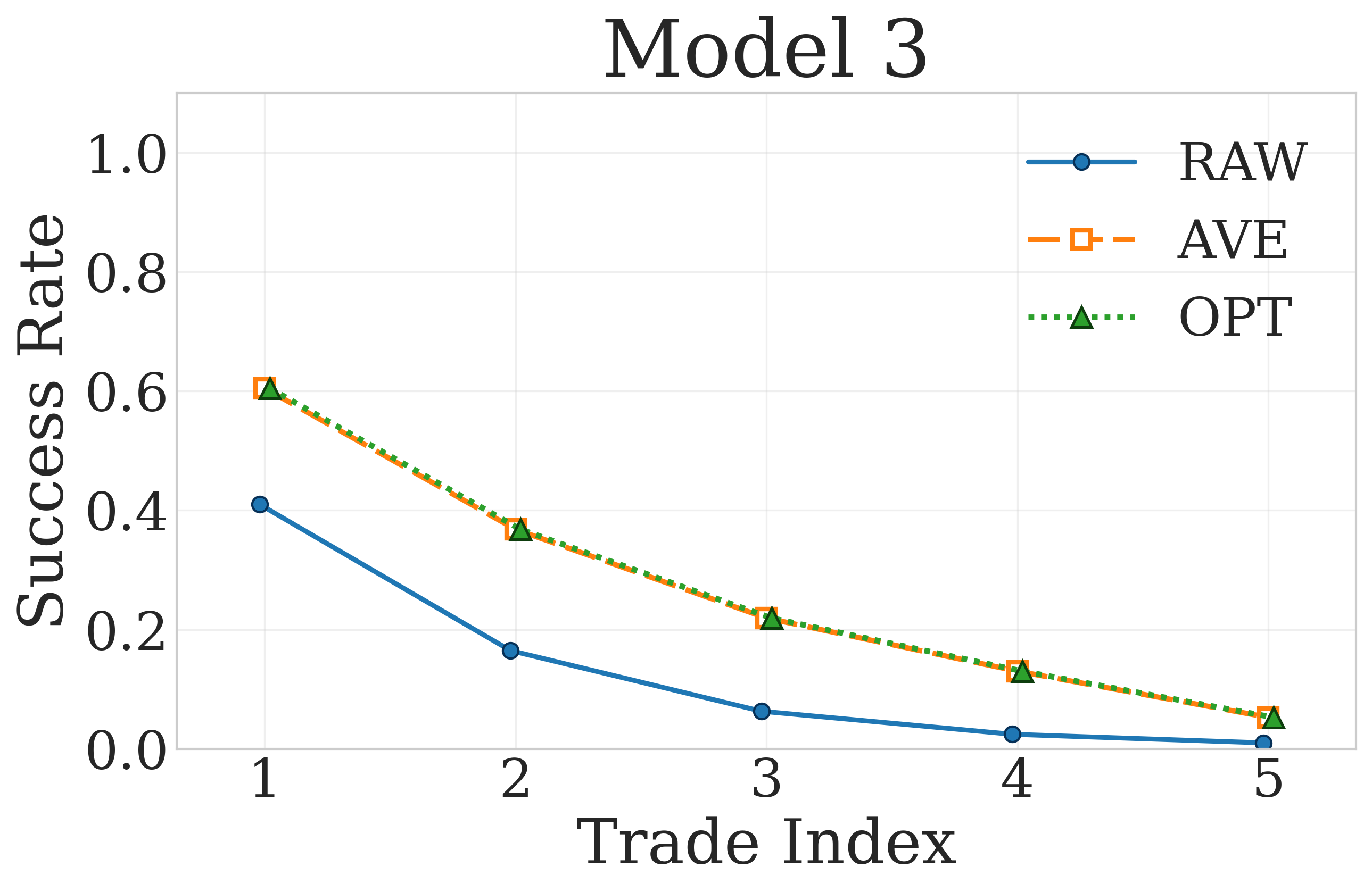}
\end{subfigure}
\hfill
\begin{subfigure}{0.235\linewidth}
\includegraphics[width=\linewidth]{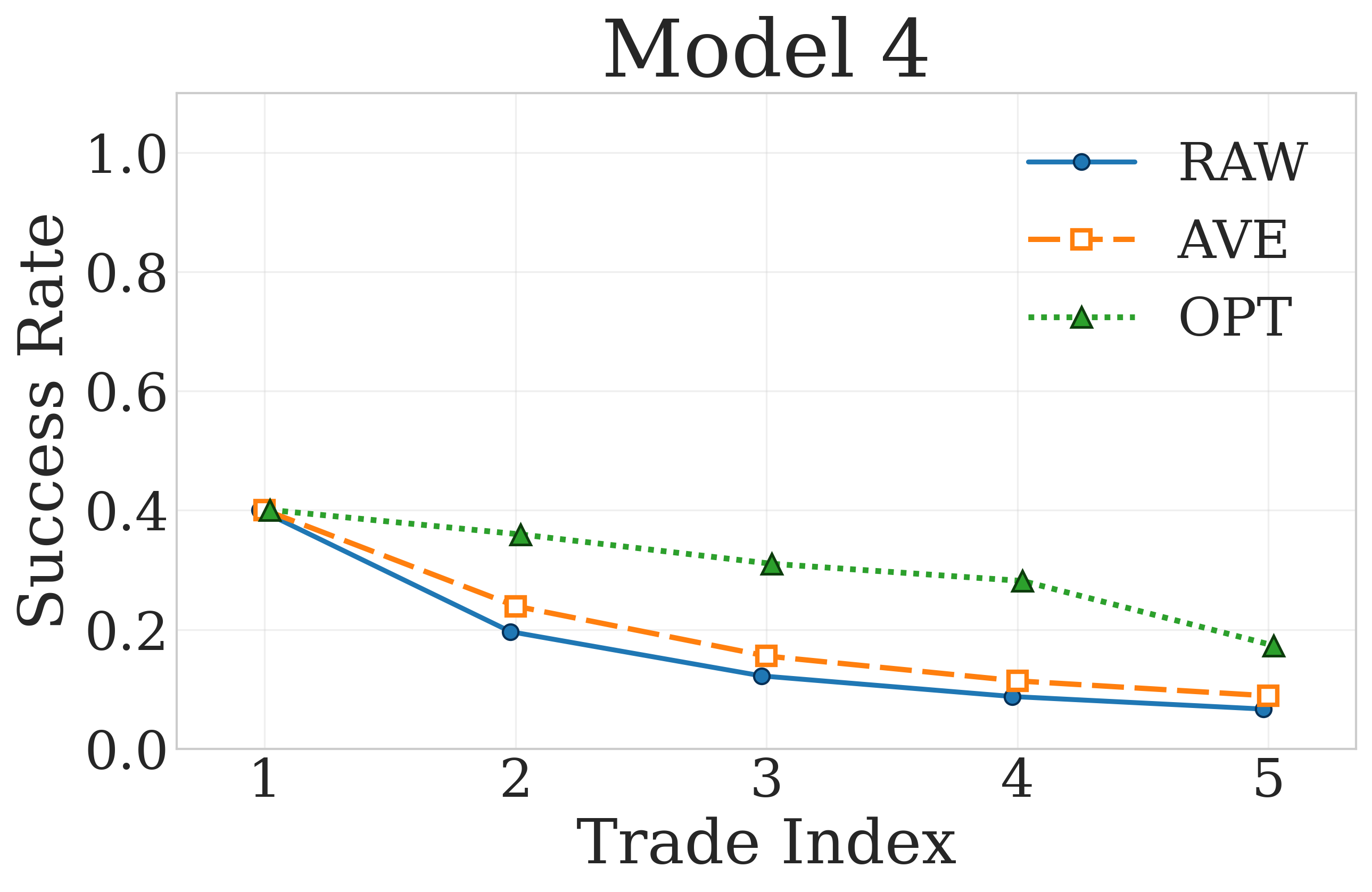}
\end{subfigure}
\medskip

\begin{subfigure}{0.235\linewidth}
\includegraphics[width=\linewidth]{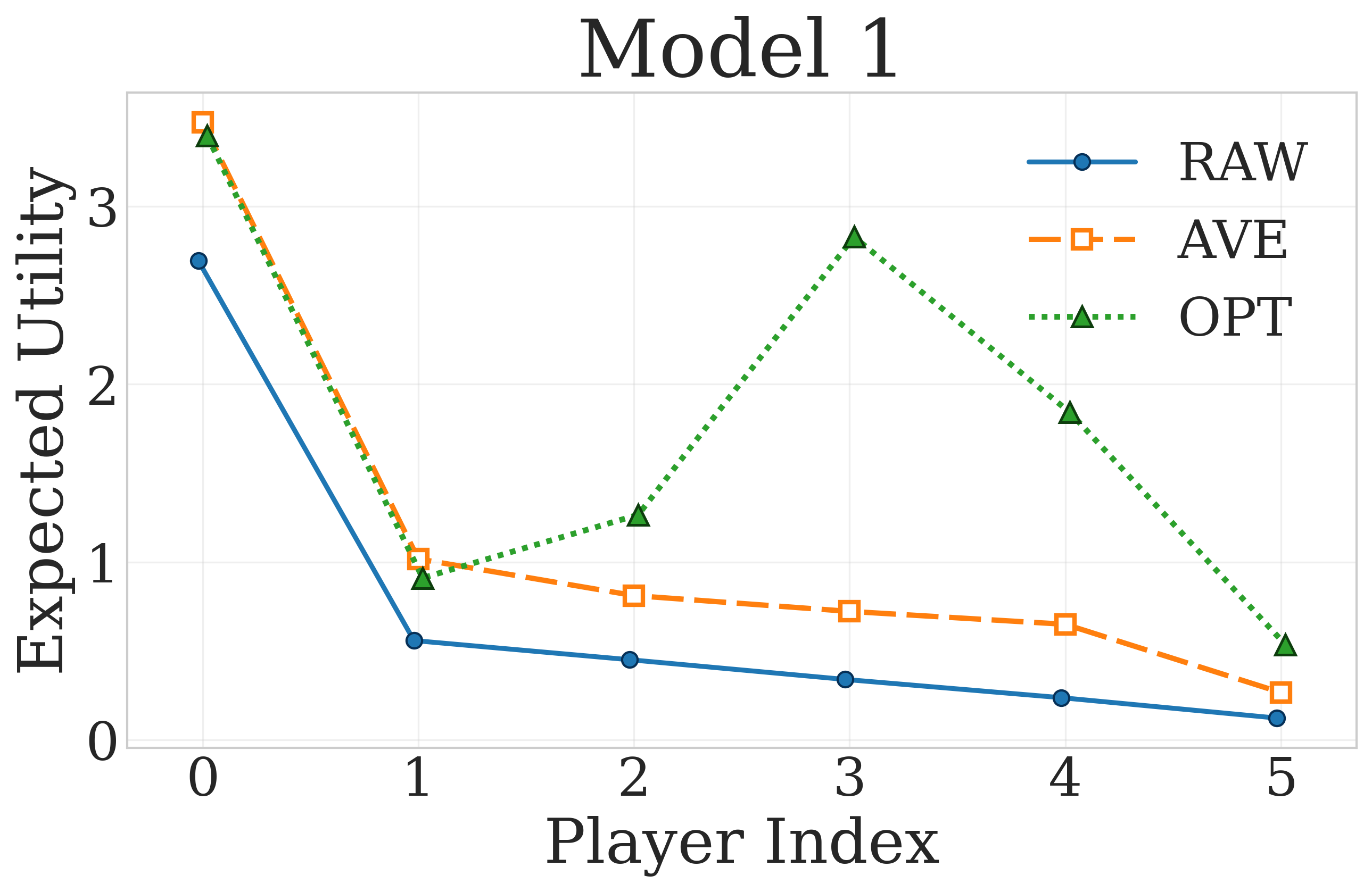}
\end{subfigure}
\hfill
\begin{subfigure}{0.235\linewidth}
\includegraphics[width=\linewidth]{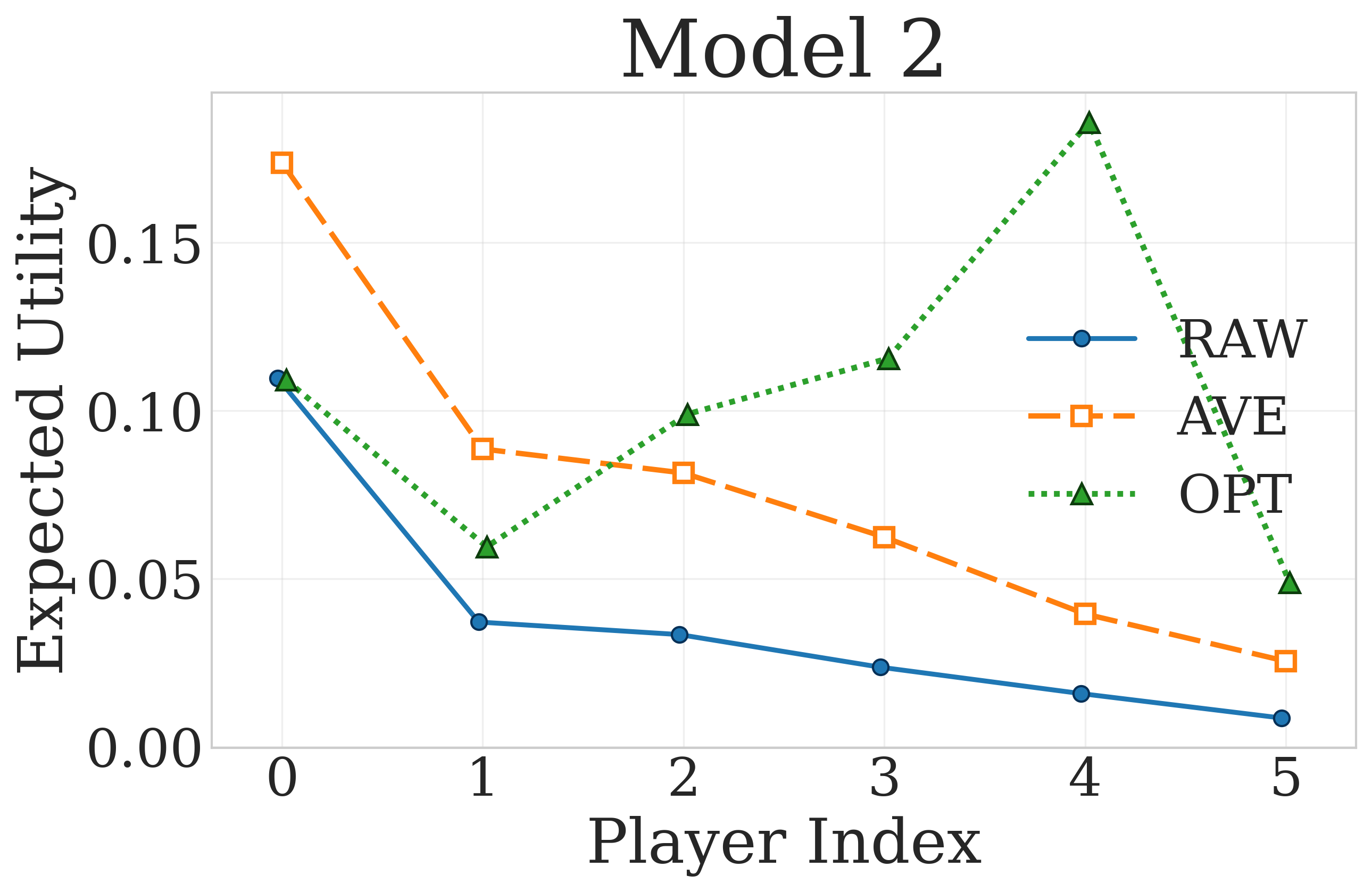}
\end{subfigure}
\hfill
\begin{subfigure}{0.235\linewidth}
\includegraphics[width=\linewidth]{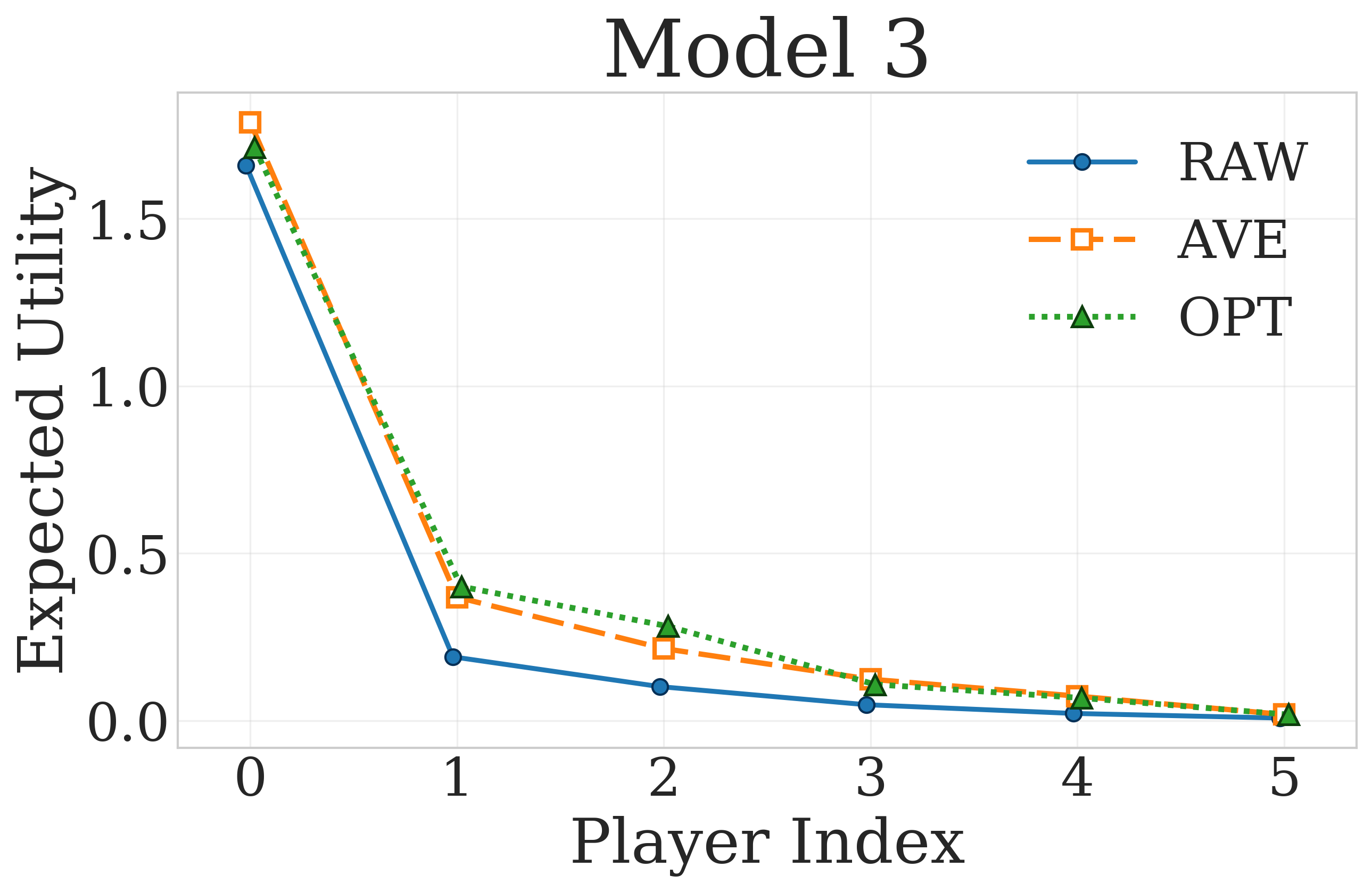}
\end{subfigure}
\hfill
\begin{subfigure}{0.235\linewidth}
\includegraphics[width=\linewidth]{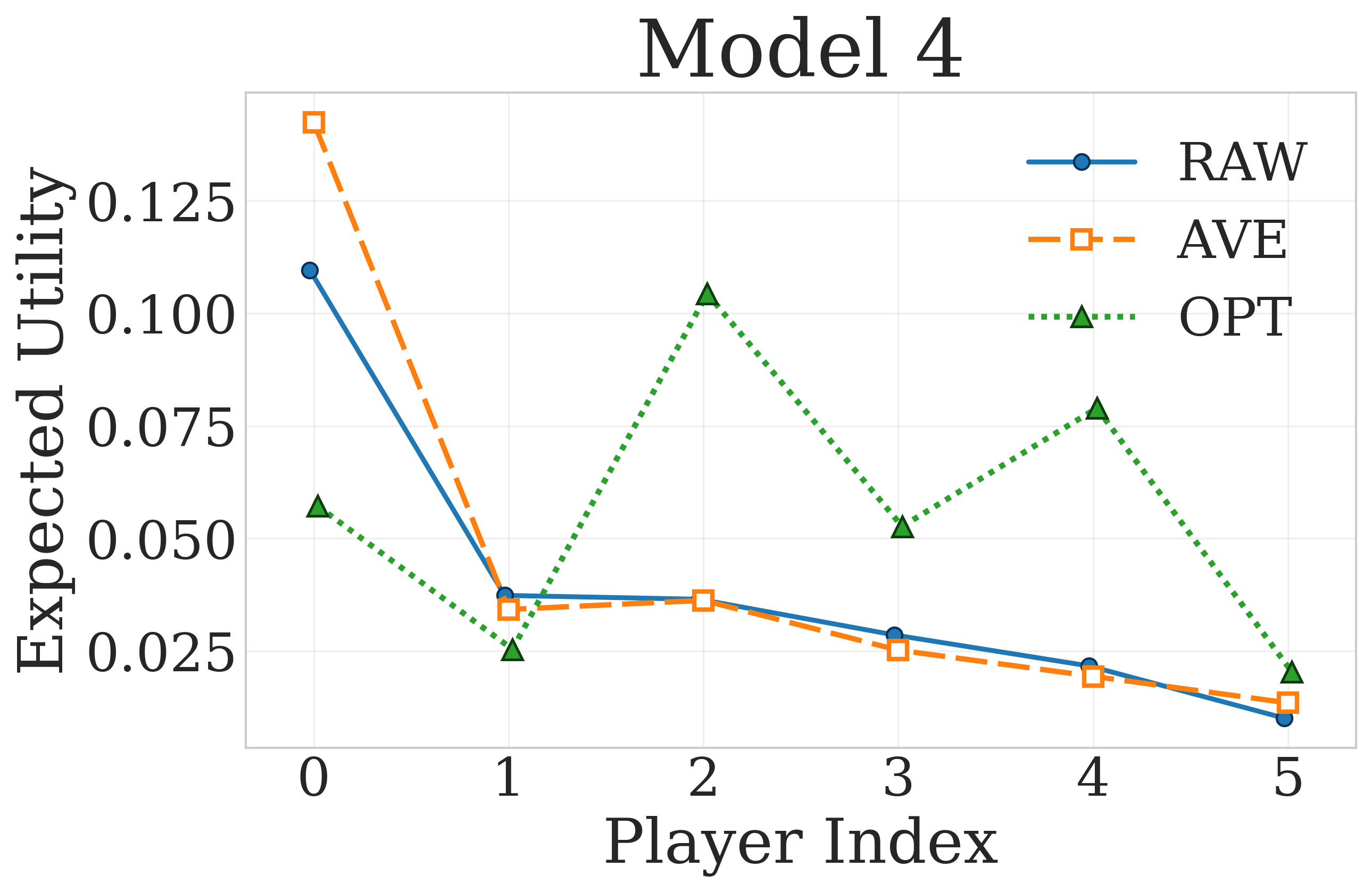}
\end{subfigure}

\caption{For each model with fixed chain length $n=5$, we plot per-trade success rates (top row) and per-player expected utilities (bottom row) under RAW, AVE, and OPT at equilibrium.}
\label{fig:exp}
\end{figure*}

\begin{figure*}[t]
\centering
\begin{subfigure}{0.235\linewidth}
\includegraphics[width=\linewidth]{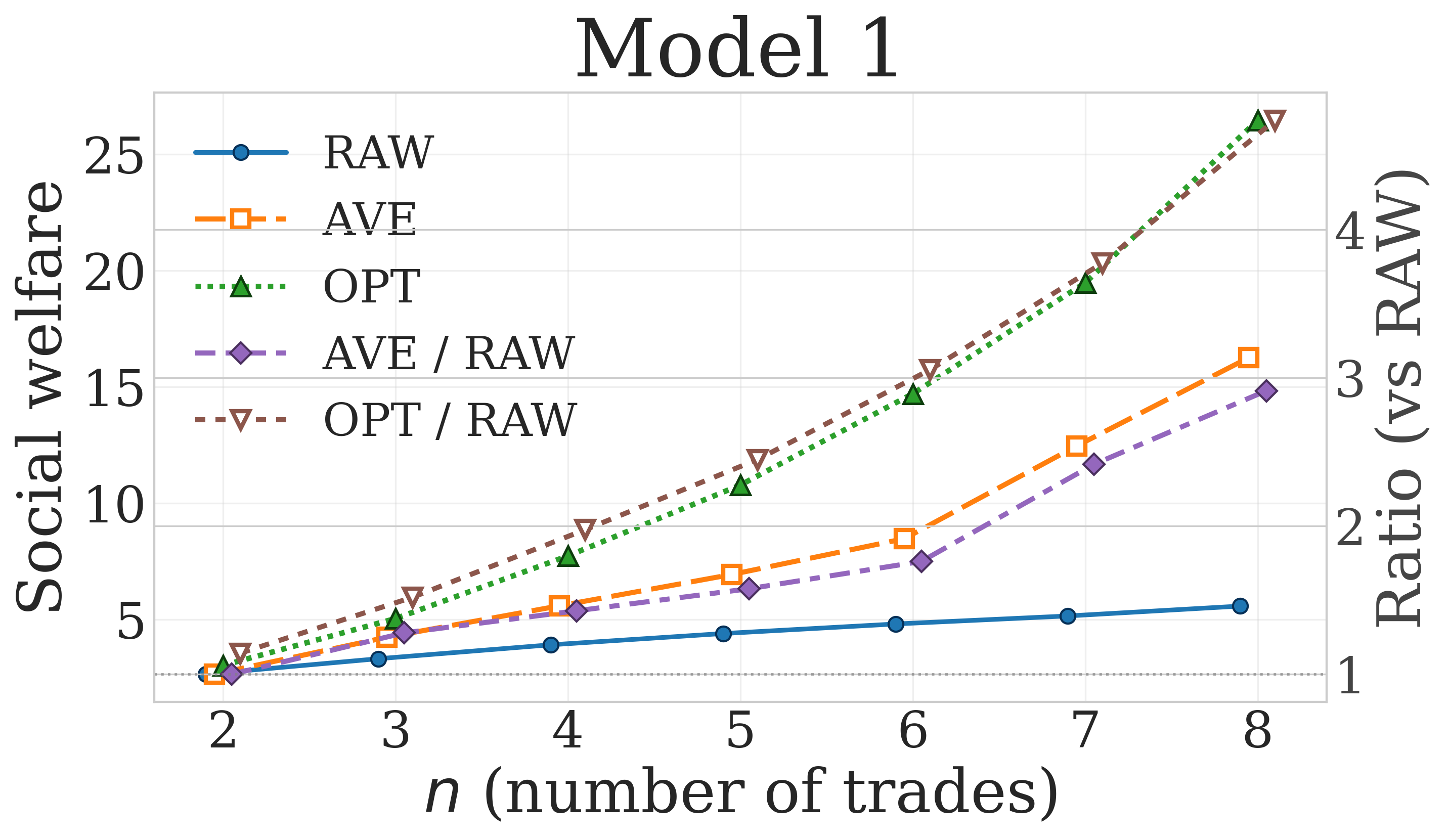}
\end{subfigure}
\hfill
\begin{subfigure}{0.235\linewidth}
\includegraphics[width=\linewidth]{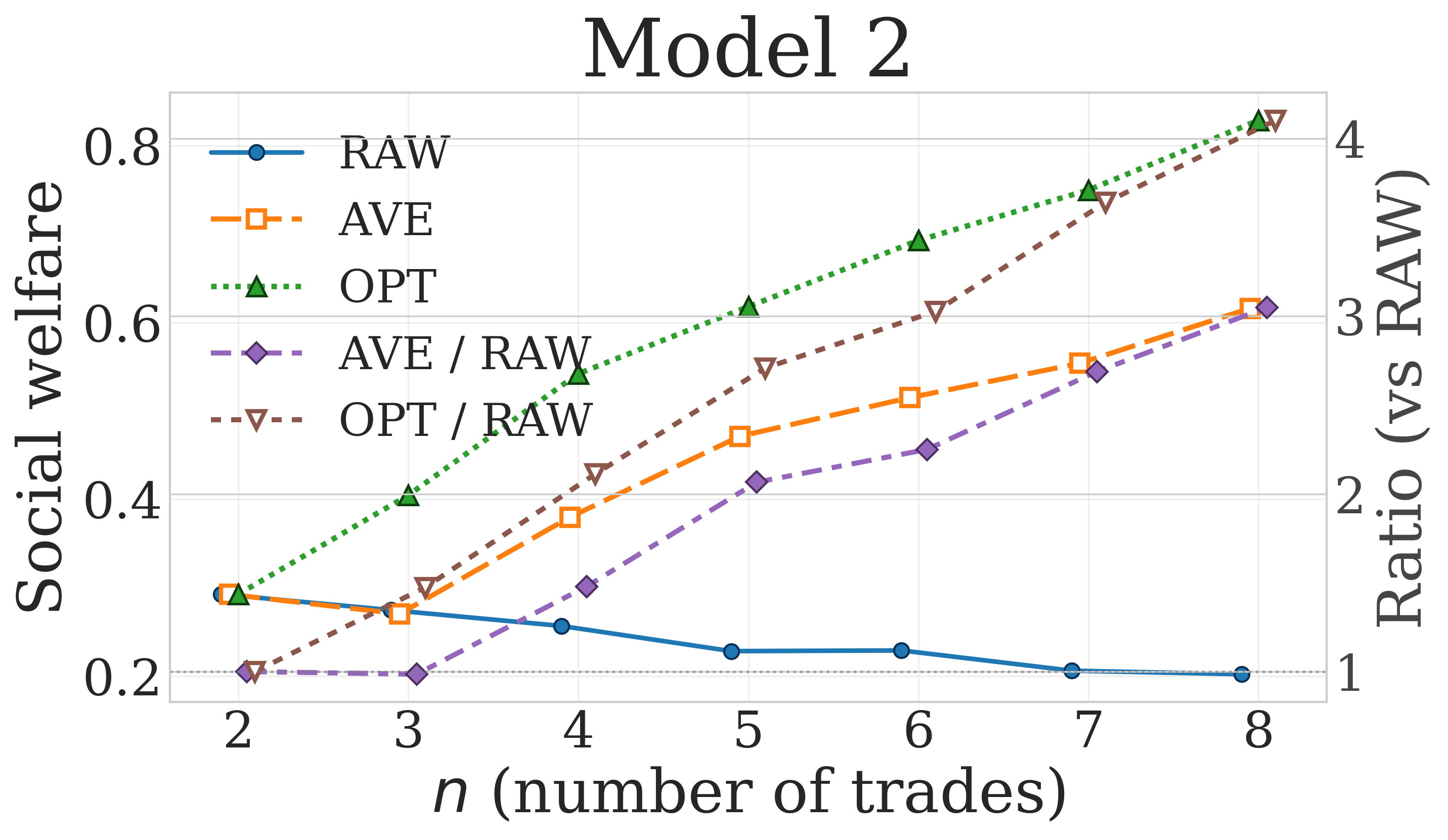}
\end{subfigure}
\hfill
\begin{subfigure}{0.235\linewidth}
\includegraphics[width=\linewidth]{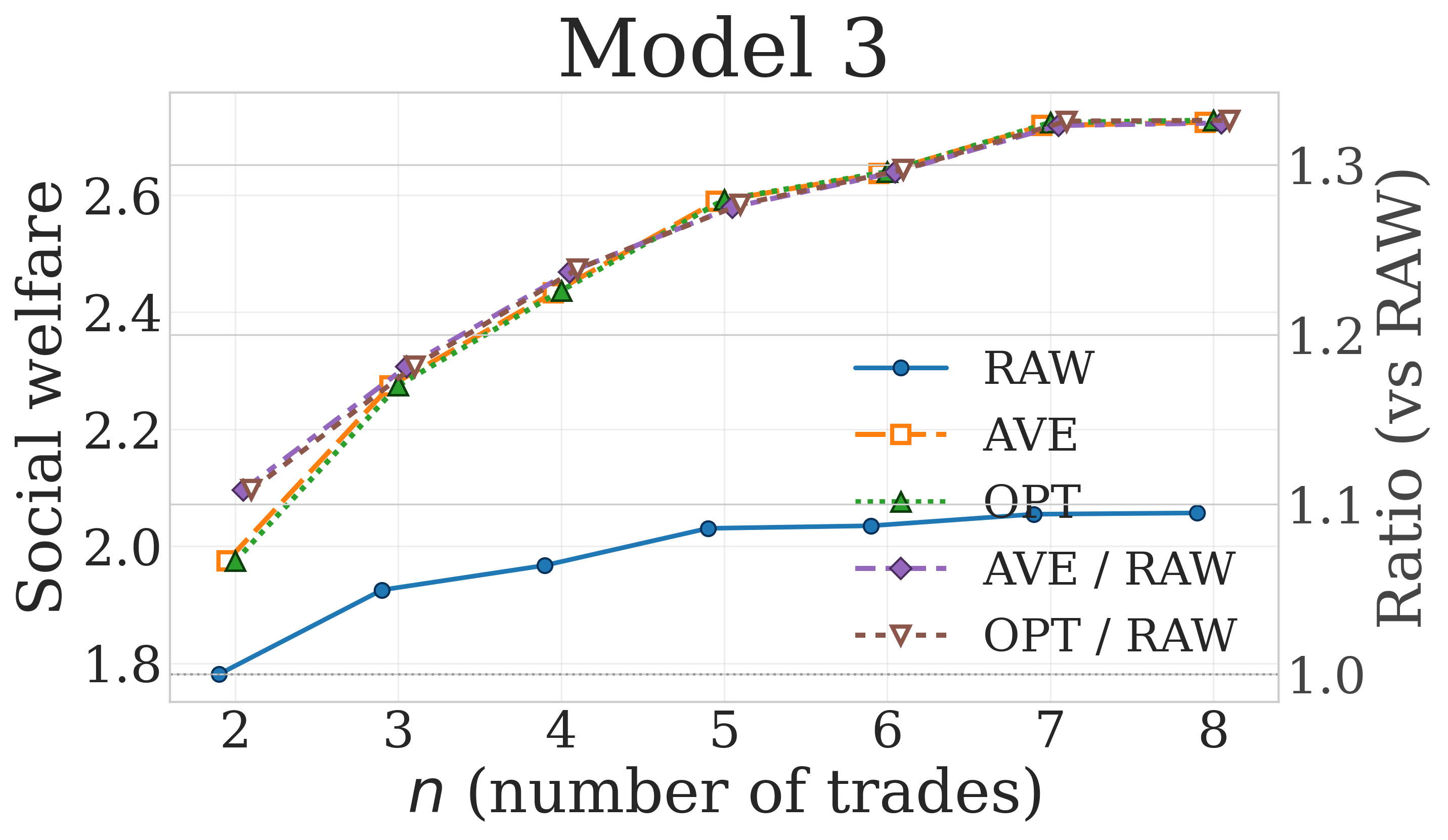}
\end{subfigure}
\hfill
\begin{subfigure}{0.235\linewidth}
\includegraphics[width=\linewidth]{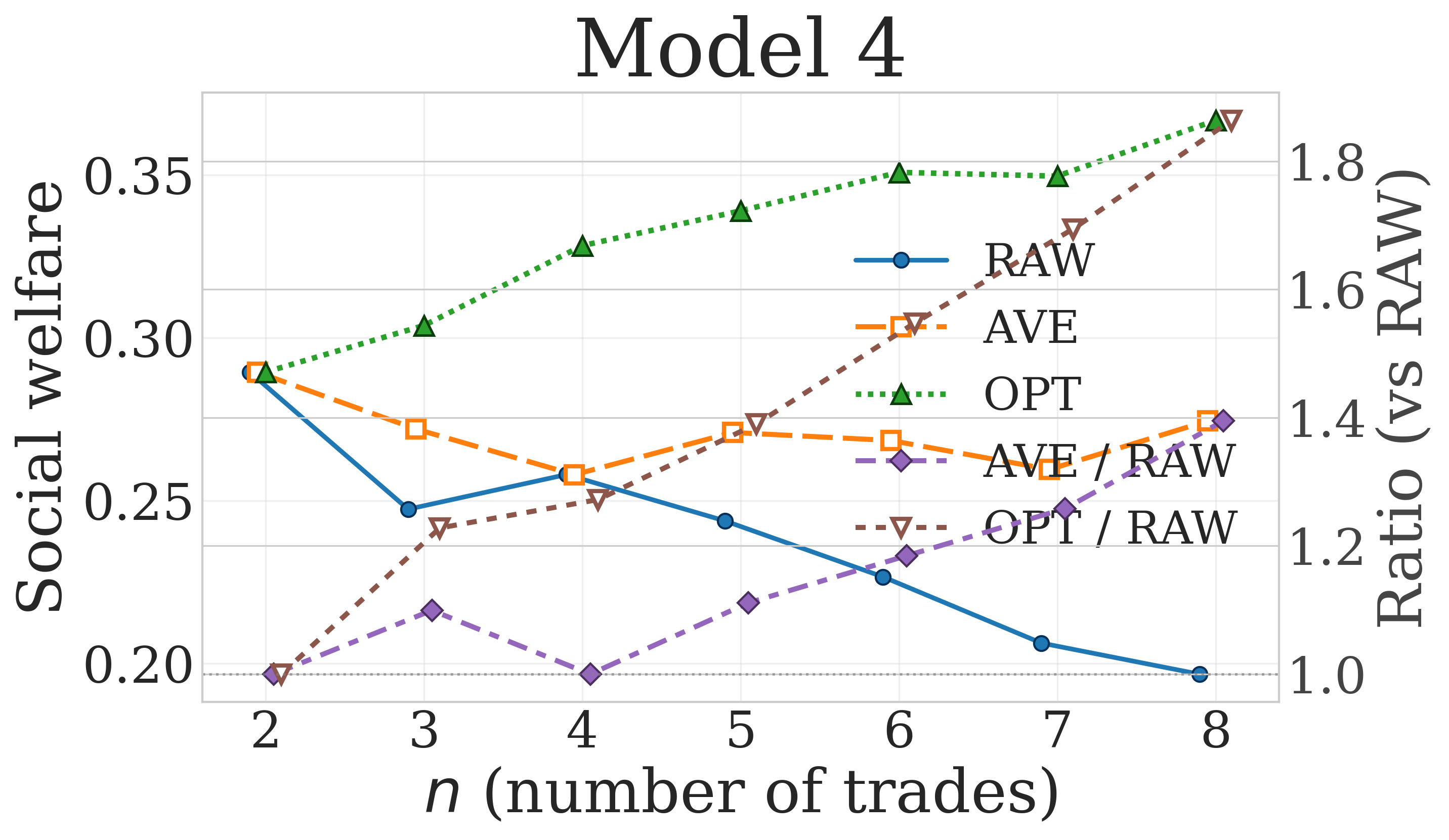}
\end{subfigure}
\medskip

\begin{subfigure}{0.235\linewidth}
\includegraphics[width=\linewidth]{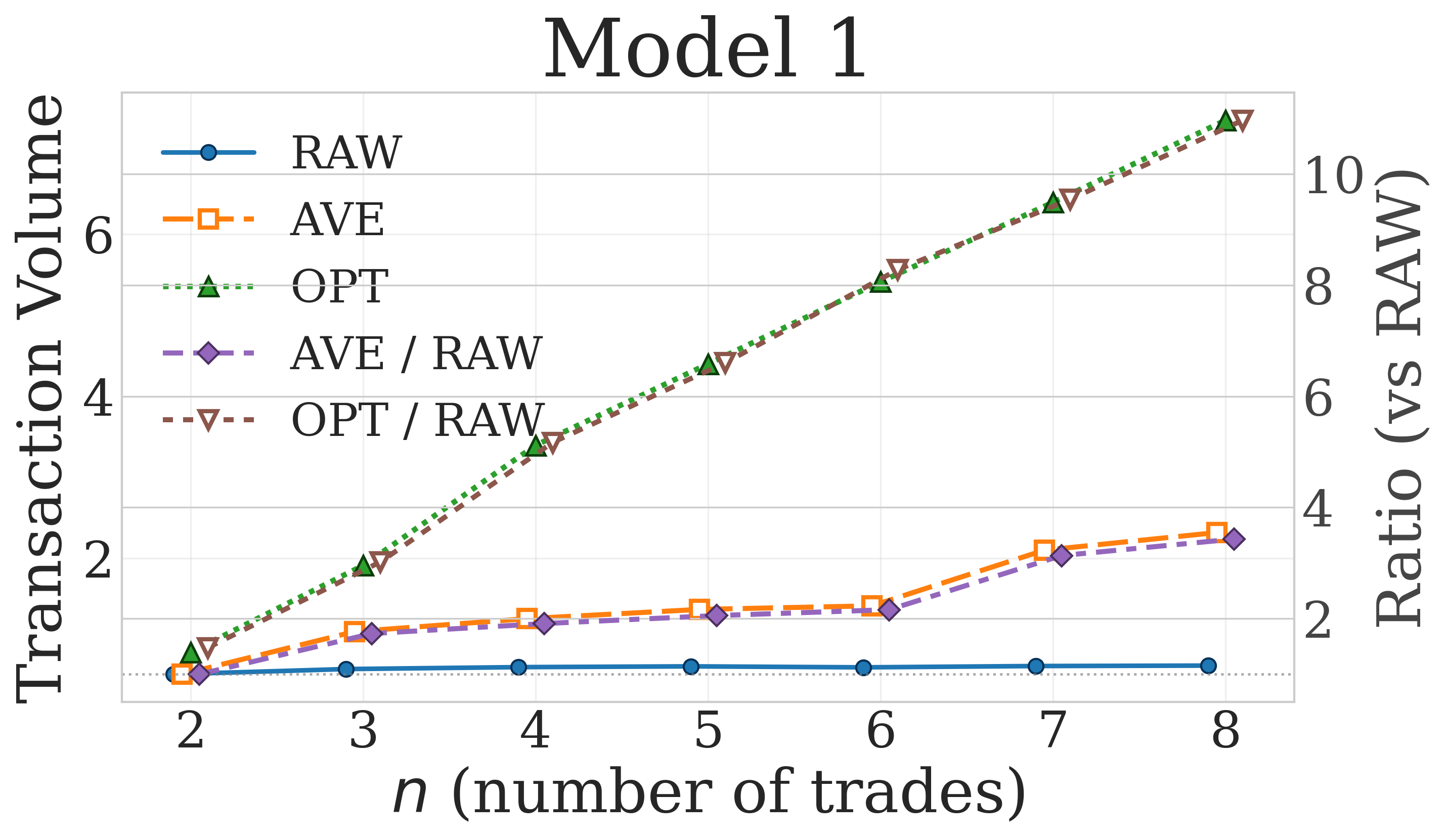}
\end{subfigure}
\hfill
\begin{subfigure}{0.235\linewidth}
\includegraphics[width=\linewidth]{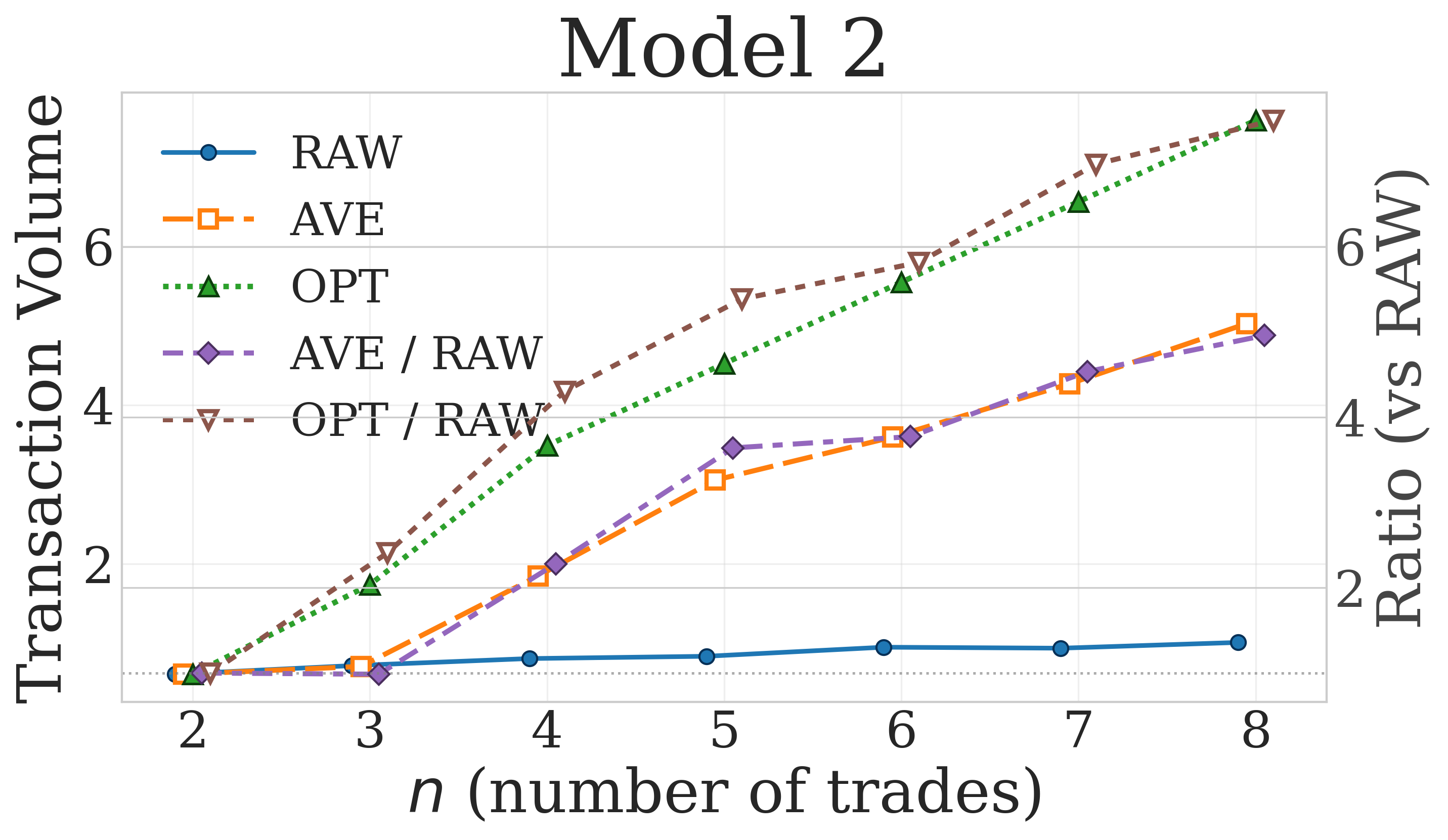}
\end{subfigure}
\hfill
\begin{subfigure}{0.235\linewidth}
\includegraphics[width=\linewidth]{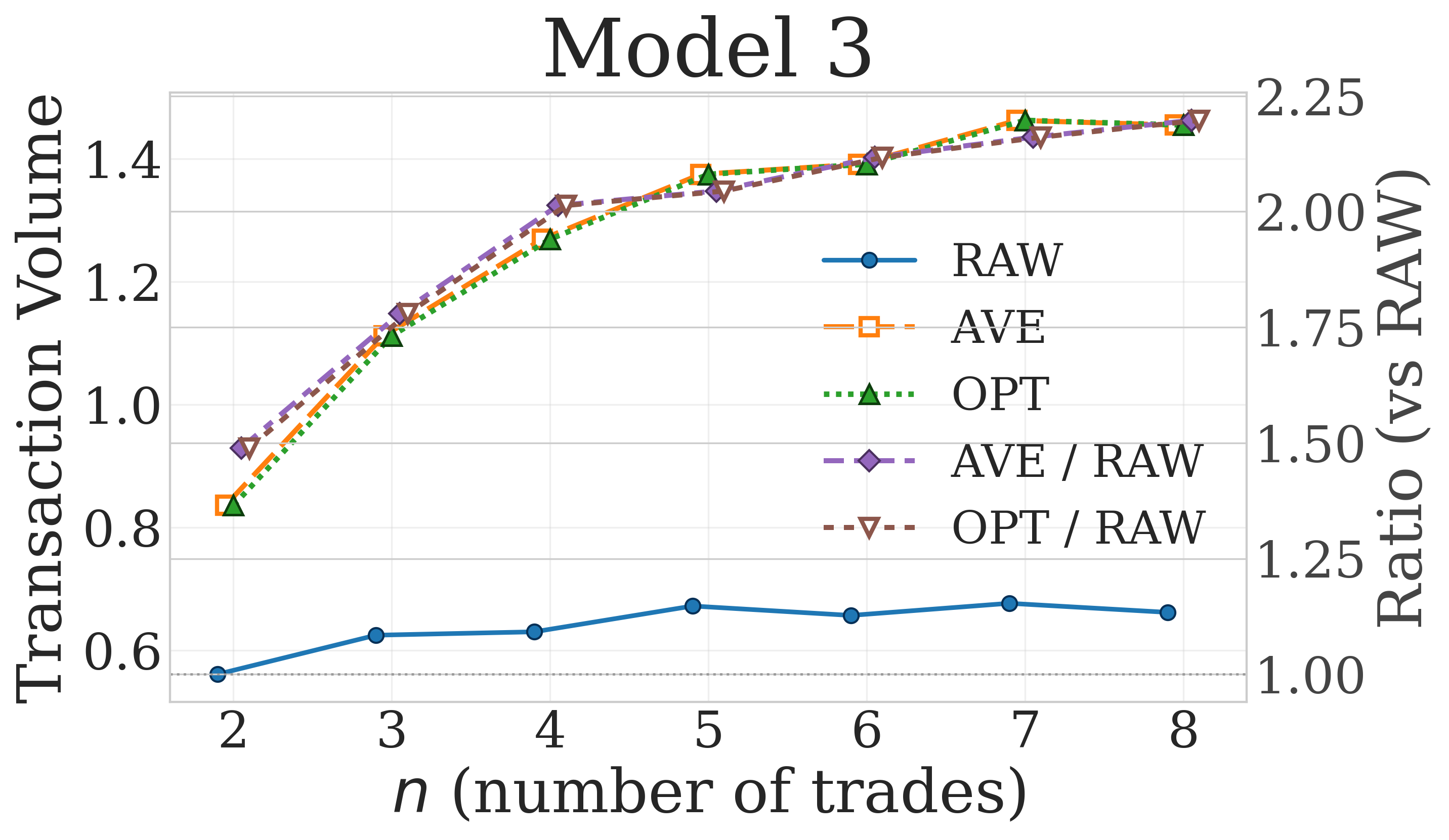}
\end{subfigure}
\hfill
\begin{subfigure}{0.235\linewidth}
\includegraphics[width=\linewidth]{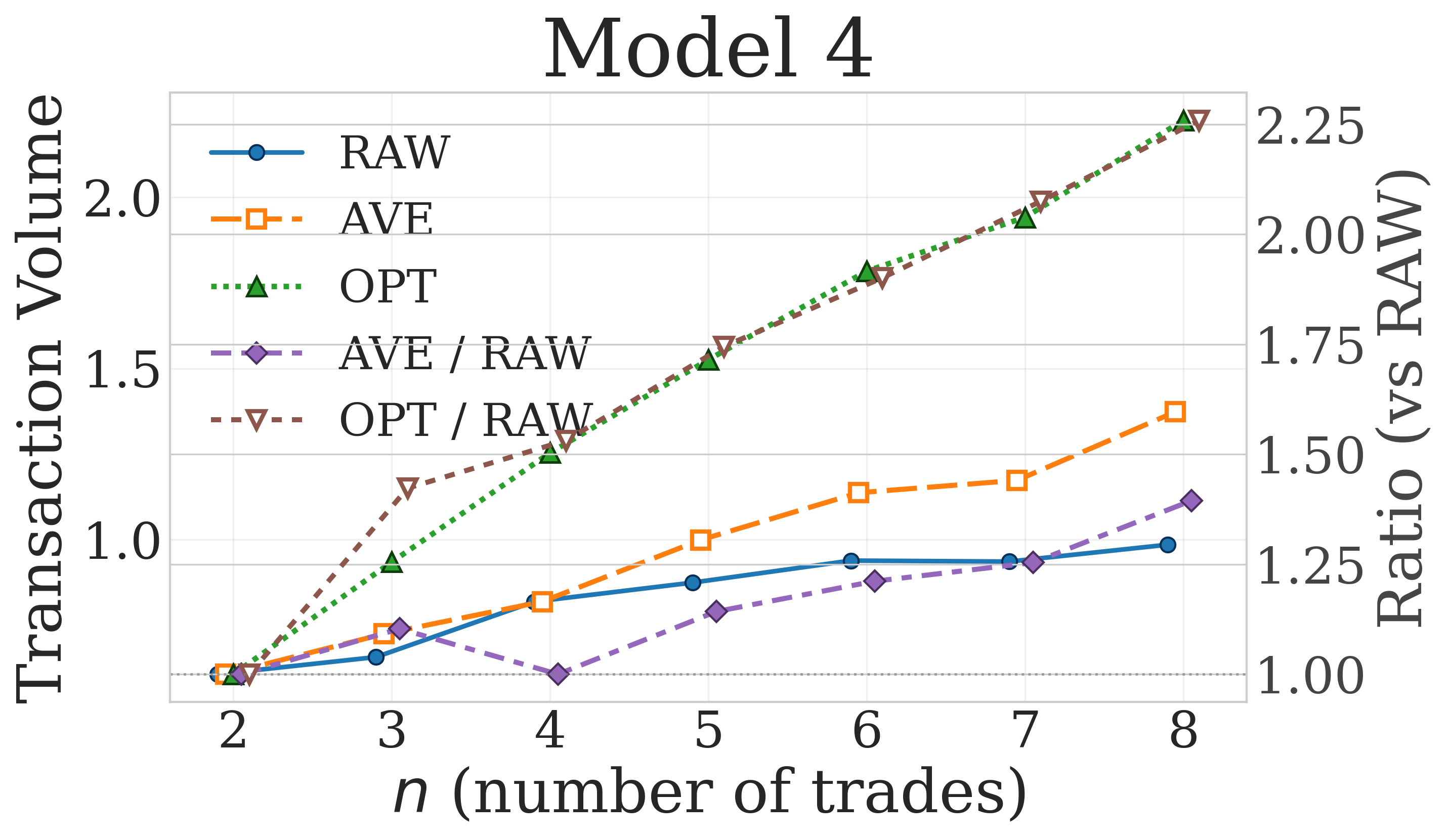}
\end{subfigure}

\caption{Expected social welfare (top row) and transaction volume (bottom row) as the chain length $n$ varies from $2$ to $8$.}
\label{fig:exp-vary}
\end{figure*}

\paragraph{Welfare improvements.}
To address questions (a) and (b), \cref{tab:exp} reports equilibrium transaction volume and social welfare for each mechanism-model pair.
Profit reallocation yields substantial improvements in both metrics, including in models that violate \cref{asp:homo,asp:posi}.
Relative to RAW, OPT improves transaction volume and social welfare by 291.9\% and 95.0\% on average, respectively.
Even the simple, model-agnostic AVE mechanism improves transaction volume and social welfare by 120.0\% and 50.4\% on average.
In addition, in the positive-valuation models, OPT induces almost all trades to occur, which accounts for a large portion of the welfare gain.

\paragraph{Robustness.}
To address question (c), we vary the chain length from $n=2$ to $n=8$ and evaluate expected social welfare and transaction volume under all four valuation models.
\Cref{fig:exp-vary} shows that the relative advantage of profit reallocation over RAW generally increases with chain length $n$.
The gains are robust to violations of homogeneity, while they become smaller when positivity is not satisfied.

\paragraph{Fairness.}
Because profit reallocation redistributes profits among players, it is important to examine whether the mechanism reduces the utility of individual players.
\Cref{fig:exp} reports the success rate of each trade and the expected utility of each player.
Relative to RAW, almost all players obtain weakly higher utility under OPT or AVE, with the exception of players $0$ and $1$ in Model 4.
This pattern suggests that the reduction in one-shot profit is typically offset by future reallocated profits and the larger transaction volume.
These results address question (d): profit reallocation generally preserves individual participation incentives.


\section{Conclusion and Future Work}
\label{sec:conclusion}

This paper presents a game-theoretic, chain-based data trading model together with a profit reallocation mechanism that aims to incentivize trades.
We develop a polynomial-time algorithm and an FPTAS for computing equilibria under the proposed mechanism.
Both theoretical analysis and empirical evaluation show that profit reallocation increases transaction volume and improves expected social welfare compared with a baseline mechanism that does not reallocate profits.

Future work includes extending the chain-based model to trees or DAGs, studying alternative solution concepts beyond sequential equilibria, such as anti-collusion and sybil-proofness, as well as extending posted-price dynamic to more general dynamics.



\newpage
\bibliographystyle{plainnat}
\bibliography{_reference}

\appendix
\crefalias{section}{appendix}
\crefalias{subsection}{appendix}
\crefalias{subsubsection}{appendix}
\newpage
\section*{Appendix}

\section{Further Related Work}
\label{sec:related}

A line of works focus on ensuring the security, privacy, trustworthiness, and tractability of data utilization through federated learning \cite{bonawitz2019towards, hard2018federated}, differential privacy \cite{abadi2016deep, liu2021dealer}, and blockchain \cite{xin2025tbds, kim2024blockchain}.
These works enable tracking transaction chains, allowing us to execute profit reallocation in real world based on this technique.

In practical data trading, data flow typically follows a specific sequential direction, constituting a chain structure.
For instance, in training of Large Language Models (LLMs), raw data on the internet may be processed to produce clean data \citep{CommonCrawl_Repository}, which are then used by LLM algorithm \citep{agarwal2019marketplace} and finally sold to users \citep{bommasani2021opportunities}.
Household data is another example: these data are often traded among data brokers and companies for better strategic decisions-making \citep{FTC2014DataBrokers,acquisti2016economics}.

Related to our study, diffusion auctions, referral auctions, and multilevel marketing mechanisms \citep{emek2011mechanisms, li2017mechanism, li2022diffusion, zhang2020redistribution} share a similar idea of rewarding intermediate participants; however, players are often rewarded for propagating information instead of propagating trades.
\citet{condorelli2017bilateral,manea2018intermediation} also study players' equilibrium behavior on sequential trades by explicitly modeling each trade process as a take-it-or-leave-it offer.
Our study differs from these works in many aspects. First, the traded object is replicable data rather than non-replicable traditional goods, so the seller can sell the data without losing it.
Second, we explicitly study the effect of external mechanisms on the players' trading behaviors.
Lastly, data trading often occurs on a platform rather than a social network; players typically do not know one another, making collusion unlikely to form between players.

\section{Omitted Algorithms}
\label{sec:omitted-algo}

\cref{sec:omitted-algo} provides all algorithms used in \cref{subsec:exact-finite,subsec:approx-continuous}.

\begin{algorithm}[!htp]
\caption{Algorithm for Exact Sequential Equilibrium}
\label{alg:exact}
\KwIn{$\mathrm{MODEL}=\left(n, L, W=\{w_l\}_\linL, \bmt = \{t_{ij}\}_{0\le i,j\le n}, Q = \{q_{i,k,l}\}_{\iinn,\kinL,\linL}\right)$}
\KwOut{$\{u^*(i,k)\}_{\iinn,\kinL}, \{p^*(i,k)\}_{0\le i<n, \kinL}$}
\textbf{Constraint: $\sum_\linL q_{i,k,l} = 1$}\\
$u^*(n,k) \gets w_k$\\
\For{$i=n-1,n-2,...,1,0$}
{
    \For{$k=1,2,...,L$}
    {
        $u^*(i,k),l^*(i,k) \gets \text{Algorithm \ref{alg:exact-OS}}(\mathrm{MODEL},i,k,$ \\
        $\qquad \{u^*(j,l),p^*(j,l),r^*(i,l)\}_{j>i,\linL})$ \\
        \eIf{$l^*(i,k) \ne -1$}
        {
            $p^*(i,k) \gets u^*(i+1,l^*(i,k))$\\
            $r^*(i,k) \gets p^*(i,k) \cdot \left( \sum_l \one\{u^*(i+1, l) \ge p^*(i,k)\} q_{i,k,l} \right)$\\
        }
        {
            $p^*(i,k) \gets 0$\\
            $r^*(i,k) \gets 0$\\
        }
    }
}
\textbf{return}: $\{u^*(i,k)\}_{\iinn,\kinL}, \{p^*(i,k)\}_{0\le i<n, \kinL}$
\end{algorithm}

\begin{algorithm}[!htp]
\caption{OptimalStrategy}
\label{alg:exact-OS}
\KwIn{$\mathrm{MODEL}, i, k, \{u^*(j,l), p^*(j,l), r^*(i,l)\}_{j>i,\linL}$}
\KwOut{$u^*(i,k), l^*(i,k)$}
\textbf{Constraints:} $0\le i < n$\\
$\calL \gets \{l\in [L]: u^*(i+1, l) \ge 0 \}$\\
\If{$\calL = \emptyset$}{\textbf{return}: $u^*(i,k) \gets w_k, l^*(i,k) \gets -1$}
\For{$k' \in \calL$}
{
    $u(i,k,k') \gets \text{Algorithm \ref{alg:exact-EU}}(\mathrm{Model}, i, k, k', \{u^*(j,l), p^*(j,l), r^*(j,l)\}_{j>i,\linL})$\\
}

$l^*(i,k) \gets \argmax_{k'\in\calL} u(i,k,k')$\\
$u^*(i,k) \gets u(i,k,l^*(i,k))$\\
\textbf{return}: $u^*(i,k), l^*(i,k)$
\end{algorithm}

\begin{algorithm}[!htp]
\caption{ExpectedUtility}
\label{alg:exact-EU}
\KwIn{$\mathrm{MODEL}, i, k, k', \{u^*(j,l), p^*(j,l), r^*(i,l)\}_{j>i,\linL}$}
\KwOut{$u(i,k,k')$}
\textbf{Constraints:} $0\le i < n$\\
\For{$l=1,2,...,L$}
{
    $s_{i,k,k'}(i+1,l) \gets \one\{u^*(i+1,l) \ge u^*(i+1,k')\} \cdot q_{i+1,k,l}$\\
}
\For{$j=i+1,...,n-1$}
{
\For{$l'=1,2,...,L$}
{
    $s_{i,k,k'}(j+1,l') \gets \sum_\linL \left( s_{i,k,k'}(j,l) \cdot q_{j+1,l,l'} \cdot \one\{u^*(j+1,l')\ge p^*(j,l)\}\right)$\\
}
}
success rate $\gets \sum_{l\in L} \left(q_{i+1,k,l} \cdot \one\{u^*(i+1,l)\ge u^*(i+1,k')\}\right)$\\
one-shot profit $\gets w_k + \left( t_{i+1,i}\cdot u^*(i+1,k') \cdot \text{success rate}\right)$\\
future reallocated profit $\gets \sum_{i<j<n,\linL} \left(s_{i,k,k'}(j,l) \cdot t_{j+1,i} \cdot r^*(j,l)\right)$\\
\textbf{return}: $\text{one-shot profit} + \text{future reallocated profit}$
\end{algorithm}

\begin{algorithm}[t]
\caption{Algorithm for $\varepsilon$-Approximate Sequential Equilibrium}
\label{alg:approx}
\KwIn{$\varepsilon$: approximation measure, $\mathrm{MODEL}=(n,\calF, \{t_{ij}\}_{0\le i,j\le n})$}
\KwOut{$\tbmp^*,\tbms^*$: an $\varepsilon$-approximate sequential equilibrium}
\textbf{Parameters:} $L_0, H_0\in \bbR, K_0\in\bbR_+$\\
\eIf{$H_0 > 0$}
{
    $\varepsilon_0 \gets \frac{\min(\varepsilon,1)}{4(1 + n K_0)(1 + \max(0,-L_0))}$
}
{
    $\varepsilon_0 \gets \varepsilon / 2$
}
$L \gets \lfloor \frac{H_0 - L_0}{\varepsilon_0} \rfloor + 1$\\
$W = \{L_0 + (l-1)\cdot \varepsilon_0\}_\linL$\\
\For{$\iinn,k,\ell\in [L]$}
{
$q(i,k,l) = F_i(w_l|w_k) - F_i(w_{l-1}|w_k)$\\
}
$\{u^*(i,k)\}_{\iinn,\kinL}, \{p^*(i,k)\}_{0\le i<n, \kinL} \gets \text{Algorithm \ref{alg:exact}}(n,L,W,\{q(i,k,l)\}_{\iinn,\kinL,\linL}, \{t_{ij}\}_{0\le i,j\le n})$\\
Transform $\{u^*(i,k)\}_{\iinn,\kinL}, \{p^*(i,k)\}_{0\le i<n, \kinL}$ into $(\tbmp, \tbms)$ by rounding inputs with \cref{eq:fptas-1}\\
\textbf{return} $\tbmp^*$, $\tbms^*$
\end{algorithm}

\newpage
\section{Omitted Proofs}
\label{sec:omitted-proofs}

\subsection{Proof of \texorpdfstring{\cref{lem:simplify-sell}}{}}
\lemSimplifySell*

\begin{proof}
\label{prf:lem:simplify-sell}
We prove this by backward induction starting from the last \emph{trade} $n$.

Recall that player $n$'s expected utility is $u_n(x_n;v_n,h_n) = \prod_{j<n} x_j \cdot x_n(v_n - p_{n-1})$ from \cref{eq:utility}.
It's clear that we only need to consider the case that trade $n-1$ occurs ($x_j = 1$ for all $j \in [n-1]$).
Then it simplifies to $u_n(x_n;v_n,h_n) = x_n(v_n - p_{n-1})$
Clearly, regardless of player $i\le n-1$'s strategies, there always exists player $n$'s optimal strategy $\tilde{x}_n^*(v_n, p_{n-1})$ that relies only on $(v_n,p_{n-1})$.

Now, we can assume that player $n$'s strategy always has the form in \cref{lem:simplify-sell}.
Then, we consider player $n-1$'s expected utility by choosing optimal $p_{i-1}$, which is:
\begin{equation}
\label{eq:prf:lem:simplify-sell:1}
\begin{aligned}
& u_{n-1}(x_{n-1}, p_{n-1};v_{n-1}, h_{n-1}) = 
\left( \prod_{j\le n-1} x_j \right) \cdot
\\
& \left( v_{n-1} - p_{n-2} + \bbE_{v_n\sim \tmu_{n-1}(\cdot | v_{n-1},h_{n-1})} [t_{n,n-1}\cdot p_{n-1}\cdot \tx_n^\dagger(v_n, p_{n-1})] \right).
\end{aligned}
\end{equation}

Note that trade $i-1$ occurs is the only interested case. In this case, player $i-1$'s utility simplifies to 
\begin{equation}
\label{eq:prf:lem:simplify-sell:1-1}
\begin{aligned}
& u_{n-1}(x_{n-1}, p_{n-1};v_{n-1}, h_{n-1}) = 
\\
& \left( v_{n-1} - p_{n-2} + \bbE_{v_n\sim \tmu_{n-1}(\cdot | v_{n-1},h_{n-1})} [t_{n,n-1}\cdot p_{n-1}\cdot \tx_n^\dagger(v_n, p_{n-1})] \right).
\end{aligned}
\end{equation}

Note that $p_{n-1}$ only appears in the expectation term in \cref{eq:prf:lem:simplify-sell:1-1}.
Also note that by \cref{def:markov}, $\tmu_{n-1}$ depends only on $v_{n-1}$.
As a consequence, the expectation term in \cref{eq:prf:lem:simplify-sell:1-1} relies only on $v_{n-1}$ and $p_{n-1}$.
Clearly, regardless of strategies in trade $j \le n-1$, there always exists player $n-1$'s optimal strategy $\tp_{n-1}^\dagger(v_{n-1})$ that maximizes player $n-1$'s utility and relies only on $v_{n-1}$. 


Now, by backward induction, we assume the existence of optimal strategies for all \emph{trade} $j > i$ that depend only on $v_j$ and $p_{j-1}$, i.e., $\tx_j(v_j,h_j) = \tx^\dagger_j(v_j, p_{j-1}) $, $\tp_{j-1}(v_{j-1},h_{j-1}) = \tp^\dagger_{j-1}(v_{j-1})$.
We study player $i-1$'s pricing strategy and player $i$'s buying strategy in trade $i$. Again, we only need to consider the case that trade $i-1$ occurs.

First, consider player $i$'s utility from choosing $x_i$ as in \cref{eq:utility}.
Based on the induction hypothesis, all $p_j$s ($\jgei$) and $x_j$ ($j>i$) within the expectation term depend only on $\bmv_{i:n}$.
Since the posterior of $\bmv_{i+1:n}$ relies solely on $v_i$ (by \cref{def:markov} again), the entire expectation term is fully a function of $v_i$.
Overall, the utility function in \cref{eq:utility} depends only on $x_i, v_i$ and $p_{i-1}$.
Clearly, regardless of strategies in trade $j < i$, there exists a optimal strategy $\tx^\dagger_i(v_i, p_{i-1})$ that relies only on $v_i, p_{i-1}$.

Second, consider player $i-1$'s utility from setting price $p_{i-1}$ when trade $i-1$ occurs:
{
\begin{equation}
\label{eq:prf:lem:simplify-sell:2}
\begin{aligned}
& u_{i-1}(x_{i-1}, p_{i-1};v_{i-1},h_{i-1}) = (v_{i-1}-p_{i-2})
+ 
\\
& \bbE_{\bmv_{i:n}\sim \tmu_{i-1}(\cdot | v_{i-1},h_{i-1})}
\left[
\sum_{i-1 < j \le n}\left( t_{j,i-1} p_{j-1} 
\prod_{i-1 < k \le j} x_{k}
\right)\right]
\end{aligned}
\end{equation}
}
Note that $p_{i-1}$ only affects the expectation term.
Also note that by \cref{def:markov}, $\tmu_{i-1}$ depends only on $v_{i-1}$. 
As a consequence, the expectation term in \cref{eq:prf:lem:simplify-sell:2} relies only on $v_{i-1}$ and $p_{i-1}$.
Clearly, regardless of strategies in trade $j \le i-1$, there always exists a optimal strategy $\tp_{i-1}(v_{i-1})$ that maximizes player $i-1$’s utility and relies only on $v_{i-1}$.
This completes the inductive step.

The strategy profile consisting optimal strategies constructed in the above proof naturally constitutes a sequential equilibrium. This completes the proof.

\end{proof}

\subsection{Proof of \texorpdfstring{\cref{lem:simplify-buy}}{}}
\lemSimplifyBuy*
\begin{proof}
\label{prf:lem:simplify-buy}
Let $\tp_i^\dagger(v_i)$ be the equilibrium pricing strategy in \cref{lem:simplify-sell}.
Similar with proof of \cref{lem:simplify-sell}, recall that player $i$'s utility is,
\begin{equation}
\begin{aligned}
u_i(x_i, p_i;v_i,h_i) = \prod_{j \le i}x_j \left(
(v_i-p_{i-1})
+ \bbE_{\bmv_{i+1:n}\sim \calF(\cdot|v_i)}
\left[
\sum_{i < j \le n} t_{ji} p_{j-1} 
\prod_{i< k \le j} x_{k}
\right]\right)
\end{aligned}
\end{equation}

Similarly, if trade $i-1$ does not occur, then player $i$ has zero utility and any action is indifferent to player $i$.
Therefore, we focus on the interested case where trade $i-1$ occurs, then, player $i$'s utility becomes,
\begin{equation}
\begin{aligned}
u_i(x_i, p_i;v_i,h_i) = x_i \left(
(v_i-p_{i-1})
+ \bbE_{\bmv_{i+1:n}\sim \calF(\cdot|v_i)}
\left[
\sum_{i < j \le n} t_{ji} p_{j-1} 
\prod_{i< k \le j} x_{k}
\right]\right)
\end{aligned}
\end{equation}

It's clear that $\tx^\dagger_i(v_i,p_{i-1}) = \one \{ \max_{p_i} u_i(1,p_i;v_i,h_i) \ge \max_{p_i} u_i(0,p_i;v_i,h_i) \}$ can maximize player $i$'s utility.
This leads to the condition:
\begin{equation}
\begin{aligned}
v_i + \max_{p_i \in \bbR_+}\quad&
\bbE_{\bmv_{i+1:n}\sim \calF(\cdot|v_i)}
\left[
\sum_{i < j \le n} t_{ji} p_{j-1} 
\prod_{i< k \le j} x_{k}
\right] \ge p_{i-1}
\\
\Leftrightarrow\quad& \ts^\dagger_i(v_i) \ge p_{i-1}
\end{aligned}
\end{equation}
In addition, setting
\begin{equation}
\tp^\dagger_i(v_i) \in \argmax_{p_i \in \bbR_+}\quad
\bbE_{\bmv_{i+1:n}\sim \calF(\cdot|v_i)}
\left[
\sum_{i < j \le n} t_{ji} p_{j-1} 
\prod_{i< k \le j} x_{k}
\right]
\end{equation}
can maximize player $i$'s utility. This completes the proof.


\end{proof}

\subsection{Proof of \texorpdfstring{\cref{thm:algo-exact}}{}}
\thmAlgoExact*

\begin{proof}
\label{prf:thm:algo-exact}

We prove \cref{thm:algo-exact} by describing how \cref{alg:exact} can compute exact equilibrium.

\paragraph{Price Simplification.}
Before describing the main algorithm, we first establish the finite search space for the seller's optimal price.
For a fixed player $i$, consider the finite set $\mathbf{P_i}$ of player $i+1$'s possible threshold values:
\begin{equation}\label{eq:price-is-next-u}
\mathbf{P_i} = \big\{\,u^*(i{+}1,\ell):\ \ell\in[L]\,\big\}.
\end{equation}

A critical observation is that $\mathbf{P_i}$ must contain player $i$'s optimal price as long as $\mathbf{P_i}$ consists of at least one non-negative element. In the opposite case that all elements in $\mathbf{P_i}$ are negative, $p_i=0$ is one of the optimal prices for player $i$.

The reason is as follows:
(a) If all elements in $\mathbf{P_i}$ are negative, whatever price $p_i \in \bbR_+$ player $i$ sets, player $i+1$ will not buy the data from player $i$. (Recall that player $i+1$ with $v_{i+1} = w_l$ will buy the data iff $u^*(i+1,l) \ge p_i$.) Therefore, player $i$ always gets $0$ profit from all subsequent trades. As a consequence, all prices $p_i$ are indifferent to player $i$, and \Wlg $p_i = 0$ is one of the optimal prices.
(b) If some elements in $\mathbf{P_i}$ are non-negative, let the maximum value in $\mathbf{P_i}$ be $u^*(i+1,l_0)$. It's obvious that setting $p_i > u^*(i+1,l_0)$ is weakly dominated by setting $p_i \le u^*(i+1,l_0)$, as player $i$ guarantees zero profit in the former case, and non-negative profit in the latter case. Therefore, there is $p^*_i \le u^*(i+1,l_0)$ be player $i$'s optimal price.
In addition, increasing $p^*_i$ up to the next value in $\mathbf{P_i}$ does not affect player $i+1$'s buying action and weakly increases player $i$'s revenue from trade $i+1$, consequently weakly increasing player $i$'s utility.
Therefore, we conclude that there is optimal $p^*_i \in \mathbf{P_i}$.
This reduces the search for $p^*(i,k)$ from a continuous space to a discrete set of size at most $L$.

\paragraph{Backward Dynamic Program.}
The high-level idea of \cref{alg:exact} is to compute $\{u^*(i,k),p^*(i,k)\}$ by backward induction from $i=n{-}1,\ldots,0$,
with $u^*(n,k) = w_k$ as the terminal condition. We also define auxiliary variable $l^*(i,k) \in \{-1\} \cup [L]$. $l^*(i,k) = -1$ means that all elements in $\mathbf{P_i}$ are negative, and consequently we have $p^*(i,k)=0$. $l^*(i,k) \in [L]$ means that player $i$ with $v_i = w_k$ achieves optimal utility by setting price $p^*(i,k) = u^*(i+1,l^*(i,k))$.

For each player $i$ and each valuation $w_k$, we first check whether $\mathbf{P_i}$ contains non-negative elements: (a) If all elements in $\mathbf{P_i}$ are negative, we return $l^*(i,k) = -1$ in \cref{alg:exact-OS}. As player $i$ does not gain profit from selling data, we must have $u^*(i,k) = w_k$.
(b) $\mathbf{P_i}$ contains non-negative elements, and we let $\calL = \{l \in [L]: u^*(i+1,l) \ge 0\}$ as the non-negative indexes.
We enumerate over the $|\calL|$ candidate prices $p_i = u^*(i+1,k')$, where $k' \in \calL$ is the price index, and define $u(i,k,k')$ as the gain function value of player $i$ when $v_i = w_k$ and player $i$ sets price $p_i = u^*(i+1,k')$:
\begin{align*}
u(i,k,k') := g_i(w_k,u^*(i+1,k')).
\end{align*}
where $g_i$ is defined in \cref{eq:lem:simplify-buy:3}.

The optimal pricing strategy is then chosen by maximizing this gain function value over $\calL$:
\begin{align*}
&l^*(i,k)\in\arg\max_{k'\in \calL} \ u(i,k,k'), \\
u^*(i,k) &= u(i,k,l^*(i,k)),\ \  
p^*(i,k) = u^*(i+1,l^*(i,k)).
\end{align*}
We also define the expected one-shot profit as an auxiliary value, which will be used in the remaining steps:
\begin{equation}\label{eq:rev-one-shot}
r^*(i,k)\ :=\ p^*(i,k)\cdot \sum_{\ell\in[L]} q_{i,k,\ell}\, 1\!\left\{\,u^*(i{+}1,\ell)\ \ge\ p^*(i,k)\,\right\}.
\end{equation}

\paragraph{Computing the Gain Function Value $u(i,k,k')$.}
The core difficulty lies in calculating $u(i,k,k')$, which, due to the introduction of the profit reallocation mechanism, involves \emph{future reallocated profits} from all future trades $j>i$.
Naively computing this term by summing over all possible future valuation profiles $\bm v_{i+1:n}$ leads to an exponential complexity.

We overcome this by decomposing $u(i,k,k')$ and using a Forward Dynamic Program.
Note that $u(i,k,k')$ decomposes into three terms: direct valuation from data, one-shot profit from subsequent trade $i+1$, and the expected reallocated profits from all future trades:
\begin{align}
&g_i(v_i = w_k, p_i = u^*(i+1,k')) \nonumber \\
&= w_k + \bbE_{v_{i+1}}\!\Big[t_{i+1,i}\,p_i\cdot 1\{\text{$x_{i+1} = 1$}\}\Big]\label{eq:ui-two-terms}\\
&\quad+ \bbE_{\bmv_{i+1:n}}\!\Bigg[\sum_{i<j<n} t_{j+1,i}\,\tp^*_j(v_j)\cdot 1\{\text{$\prod_{i+1 \le k \le j+1} x_k = 1$}\}\Bigg].\label{eq:reallocated-profit-stream}
\end{align}

The first two terms \eqref{eq:ui-two-terms} are easy to compute within polynomial time.
The challenge is to compute the \emph{reallocated profits} term \eqref{eq:reallocated-profit-stream}. We first rewrite the reallocated profits term as a sum over the discrete events:
\begin{align}
&\bbE_{\bmv_{i+1:n}}\!\Bigg[\sum_{i<j<n} t_{j+1,i}\,\tp^*_j(v_j)\cdot 1\{\text{$\prod_{i+1 \le k \le j+1} x_k = 1$}\}\Bigg]\nonumber\\
&=\sum_{i<j<n}\sum_{\ell\in[L]} t_{j+1,i}\,p^*(j,\ell)\cdot \Pr\!\big[X(j,\ell)\big],\label{eq:reallocated-profits-expand}
\end{align}
where $X(j,l)$ is the event that $v_j = w_l$ and $\prod_{i+1 \le k \le j+1} x_k = 1$.

\paragraph{Forward Dynamic Program.}
To calculate the probabilities of event $X(j,l)$, we introduce a auxiliary event, $Y(j, l)$, defined as $\one\{v_j = w_l \text{ and } \prod_{i+1 \le k \le j} x_k = 1\}$.
We next design \emph{Forward Dynamic Program} to compute $\Pr[Y(j,l)]$ recursively:
\begin{itemize}[left=0em]
\item Initialization (Trade $i+1$): For all $l \in [L]$, the probability that $x_{i+1}=1$ and $v_{i+1}=w_l$, given $v_i = w_k$ is: \begin{equation}\label{eq:init-Y}
\Pr\!\big[Y(i{+}1,\ell)\big]
\;=\; q_{i+1,k,\ell}\cdot 1\!\left\{u^*(i{+}1,\ell)\ge u^*(i{+}1,k')\right\}.
\end{equation}
\item Recursion (Trade $j \rightarrow j+1$):
We can calculate the probability of $\Pr\!\big[Y(j{+}1,\ell')\big]$ by summing over the $L$ possible valuations $w_l$ at stage $j$:
\begin{equation}
\begin{aligned}
&\Pr\!\big[Y(j{+}1,\ell')\big]
=\Pr\left[v_{j+1} = w_{\ell'}, \prod_{i+1 \le k \le j+1} x_k = 1\right]
\\
= & \sum_l \Pr\left[v_{j+1} = w_{\ell'}, v_j = w_{\ell}, \prod_{i+1 \le k \le j+1} x_k = 1\right]
\\
= & \sum_l \Pr[v_{j+1} = w_{\ell'}, x_{j+1} = 1, Y(j,\ell)]
\\
= & \sum_l \Pr[v_{j+1} = w_{\ell'}, x_{j+1} = 1 | Y(j,\ell)] \cdot \Pr[Y(j,\ell)]
\\
= & \sum_l \Pr[v_{j+1} = w_{\ell'}, x_{j+1} = 1 | v_j = \ell] \cdot \Pr[Y(j,\ell)]
\\
= & \sum_{\ell\in[L]} \Pr\!\big[Y(j,\ell)\big]\cdot q_{j+1,\ell,\ell'}\cdot 1\!\left\{u^*(j{+}1,\ell')\ge p^*(j,\ell)\right\}.
\end{aligned}
\end{equation}
\end{itemize}
This recursion is precisely the update at lines 5--9 of \Cref{alg:exact-EU} (where $s_{i,k,k'}(j,l) \coloneqq \Pr[Y(j,l)]$).
The forward dynamic program runs from $j=i+1$ to $n-1$ in $O(n L^2)$ time, eliminating the exponential complexity through enumeration.

Finally, the required $\Pr[X(j,l)]$ is related to the expected one-shot profit $r^*(j,l)$ and $\Pr[Y(j,l)]$.
Note that conditioned on $v_j=w_\ell$, trade $j{+}1$ occurs if and only if (i) trade $j$ occurs, and (ii) player $j+1$ with valuation $v_{j+1}$ buys the data with price $p_j$. Therefore
{
\begin{equation}
\label{eq:X-from-Y}
\begin{aligned}
\Pr[X(j,\ell)|Y(j,\ell)]
=& \Pr[x_{j+1}=1 | v_j = w_l]
\\
=& \sum_{\ell'\in[L]} \Pr[x_{j+1}=1, v_{j+1} = w_{l'} | v_j = w_l]
\\
=& \sum_{\ell'\in[L]} \Pr[v_{j+1} = w_{l'}, u^*(j+1,l') > p^*(j,l) | v_j = w_l]
\\
=& \sum_{\ell'\in[L]} q_{j,\ell,\ell'}\cdot 
1\left\{u^*(j+1,\ell')\ge p^*(j,\ell)\right\}
\end{aligned}
\end{equation}
}
Combining \cref{eq:reallocated-profits-expand} and \cref{eq:X-from-Y}, we have
\begin{equation}\label{eq:rev-identity}
p^*(j,\ell)\cdot \Pr\!\big[X(j,\ell)\big] \;=\; r^*(j,\ell)\cdot \Pr\!\big[Y(j,\ell)\big].
\end{equation}
\cref{eq:rev-identity} allows the polynomial-time computation of the entire reallocated profit term in \cref{eq:reallocated-profits-expand} based on the values of $\Pr[Y(j,\ell)]$s.
This computation is implemented in line 12 of \Cref{alg:exact-EU}.



Overall, by nesting the Forward DP within the Backward DP, \cref{alg:exact} successfully computes all $p^*(i,k)$ and $u^*(i,k)$ values.
Based on \cref{lem:simplify-buy} and by defining $\tp_i^*(w_k) = p^*(i,k)$, $\ts_i^*(w_k) = u^*(i,k)$, we know the strategy profile $\tbmp^*,\tbmp^s$ constitutes an exact sequential equilibrium.

The running time is:
$$\underbrace{\mathcal{O}(n)}_{\text{Backward DP}} \times \underbrace{\mathcal{O}(L)}_{\text{Seller valuation}} \times \underbrace{\mathcal{O}(L)}_{\text{Price candidates}} \times \underbrace{\mathcal{O}(n L^2)}_{\text{Forward DP}} = \mathcal{O}(n^2 L^4),$$

\end{proof}

\subsection{Proof of \texorpdfstring{\cref{thm:algo-approx}}{}}
\thmAlgoApprox*
\begin{proof}
\label{prf:thm:approx-algo}

We first show the proof idea of \cref{thm:algo-approx}.
The algorithm consists of three steps: (i) discretize the continuous valuation support and construct a discretized model, (ii) compute the exact sequential equilibrium of the discretized model by \cref{alg:exact}, and (iii) lift those strategies with discrete inputs back to strategies with continuous inputs.
The core difficulty lies in proving a polynomial approximation bound of the output strategy profile with the discretization precision $\varepsilon_0$, the Lipschitz constant $K_0$, as well as the interval length $H_0 - K_0$.

\noindent\textbf{Interval Discretization.}
We select a discretization precision $\varepsilon_0 >0$ (to be determined later) and define the grid size $L$ and the discretized grid points $W$:
\[
L ~:=~ \Big\lfloor \varepsilon^{-1} (H_0 - L_0) \Big\rfloor + 1,\qquad
w_\ell ~:=~ \ell \cdot \varepsilon + L_0,\quad 0\le \ell\le L.
\]
The grid points $w_\ell$ partition the continuous support $[L_0,H_0]$ into cells $I_\ell = [w_\ell, w_{\ell+1})$.
We can verify that $w_\ell \in [L_0, H_0]$ for $0 \le \ell \le L-1$.

To construct the induced discretized model, we define the one-step transition probabilities $q_{i,k,\ell}$ between grid points $w_k$ and $w_\ell$ as the probability of the next player's valuation falling in $I_\ell$ given the current player's valuation is exactly $w_k$:
\[
q_{i,k,\ell} ~:=~ F_i(w_{\ell+1}\mid w_k) - F_i(w_\ell\mid w_k),
\qquad \iinn, 0\le k,\ell < L
\]
with the conventions $F_i(w_0\mid w_k)=0$ and $F_i(w_L\mid w_k)=1$. Also denote $\hat{\calF}$ as the valuation profile distribution induced by discrete transition probabilities $Q \coloneqq \{q_{i,k,\ell}\}_{\iinn, 0\le k,\ell < L}$. 

Applying the exact algorithm (\cref{alg:exact}) to the discretized model $\mathrm{Model} = (n,L,W,\bmt, Q)$ yields the exact sequential equilibrium for the discretized model, characterized by the elements:
\[
\{u^*(i,k)\}_{i\in[n],\,0\le k < L}, \qquad
\{p^*(i,k)\}_{0\le i<n,\,0\le k < L}
\]
with $p^*(i,k) = u^*(i+1,l^*(i,k))$ if $l^*(i,k) \ne -1$ and $p^*(i,k) = 0$ otherwise. 

\vspace{0.2cm}
\paragraph{Lifting to Continuous Strategies.}

The next step is to construct an approximate equilibrium $(\tbmp^*,\tbms^*)$ based on the $\{u^*(i,k)\}_{i\in[n],\,0\le k < L},\;\{p^*(i,k)\}_{0\le i<n,\,0\le k < L}$.

For any continuous valuation $v_i \in [L_0, H_0]$, we define its corresponding grid index $k(v_i)$:
\[
k(v_i) := \big\lfloor \varepsilon^{-1} (v_i - L_0)\big\rfloor,
\]
We can verify that $0 \le k(v_i) \le L - 1$ for all $\iinn, v_i\in [L_0, H_0]$.
so that $w_{k(v_i)}\le v_i < w_{k(v_i)+1}$. 
The \emph{lifted strategies} for the continuous model representation $\mathrm{Model} = (n,\calF,\{t_{ij}\}_{0\le i,j \le n})$ are piecewise constant, using the discrete solutions $u^*(i,k)$ and $p^*(i,k)$ evaluated at the corresponding grid index $k(v_i)$:
\begin{equation}\label{eq:fptas-1}
\tilde s_i^*(v_i) ~:=~ u^*\bigl(i,k(v_i)\bigr), \qquad
\tilde p_i^*(v_i) ~:=~ p^*\!\bigl(i,\,k(v_i)\bigr).
\end{equation}

Importantly, for a similar reason in the proof of \cref{thm:algo-exact} (\cref{prf:thm:algo-exact}), since player $i+1$'s threshold strategy ($\ts^*_{i+1}$) is piecewise constant over the grid cells, the optimal deviated price $p'_i$ for player $i$ can, without loss of generality, again be restricted to the finite set of potential threshold values of player $i+1$:
$
\calP_i ~:=~ \bigl\{\,u^*(i{+}1,\ell):~ 0\le \ell < L\,\bigr\}.
$
This is crucial for subsequent analysis, since we need to consider the deviations to a finite set $\calP_i$.

\vspace{0.2cm}
\noindent\textbf{The Polynomial Error Bound.}
The core challenge is giving a polynomial error bound between $\ts^*_i(v_i)$ and $\ts^\dagger_i(v_i)$ as defined by \cref{def:eps-equilibrium}.

Recall the player $i$'s gain function:
\begin{equation}\label{eq:fptas-ui}
\begin{aligned}
& g_i(v_i,p_i) = v_i ~+~ \bbE_{\bmv_{i+1:n}\sim \calF(\cdot \mid v_i)} \!\Big[\,r_i\big(p_i,\bmv_{i+1:n}\big)\,\Big],
\\
& r_i\big(p_i,\bmv_{i+1:n}\big) = \sum_{i<j\leq n}\left( t_{ji} p_{j-1} \prod_{i<k\leq j}x_k\right)
\end{aligned}
\end{equation}
where $r_i(p_i,\bmv_{i+1:n})$ is the profits from all subsequent trades $(i+1:n)$ when valuation profile is $\bmv_{i+1:n}$, player $i$ sets the price $p_i$, and players $(i+1:n)$ follow the strategies $\big(\ts_j(v_j), \tp_j(v_j)\big)$.
Note that the strategies for players $(i+1:n)$ are constant within each grid cell. Whenever $v_j\in [w_{k_j},w_{k_j+1})$ for all $j>i$, we have
\begin{equation}\label{eq:cell-constant}
r_i\big(p_i,\bmv_{i+1:n}\big) ~=~ r_i\big(p_i, w_{k_{i+1}},\ldots,w_{k_n}\big).
\end{equation}


Using \eqref{eq:cell-constant}, we can decompose the expectation term in \eqref{eq:fptas-ui} by the events
\[
W_j(k)~:=~\{\,v_j\in[w_{k},w_{k+1})\,\}\quad\text{and}\quad
\hat W_j(k)~:=~\{\,v_j=w_{k}\,\},\qquad  \forall i+1\le j \le n,\quad 0\le k \le L-1
\]
to obtain
\begin{equation}\label{eq:cell-sum}
\begin{aligned}
g_i(v_i,p_i) =& v_i + \sum_{\forall i+1 \le j \le n, 0\le k_j\le L-1} \bbE_{\bmv_{i+1:n}\sim \calF(\cdot \mid v_i)} \Big[r_i\big(p_i,\bmv_{i+1:n}\big)\one\{\forall i+1 \le j \le n, v_j \in [w_{k_j}, w_{k_j+1}) \}\Big]
\\
=& v_i + \sum_{\forall i+1 \le j \le n, 0\le k_j\le L-1} \bbE_{\bmv_{i+1:n}\sim \calF(\cdot \mid v_i)} \Big[r_i\big(p_i,w_{k_{i+1}},...,w_{k_n}\big)\one\{\bigcap_{j=i+1}^n W_j(k_j) \}\Big]
\\
=& v_i + \sum_{\forall i+1 \le j \le n, 0\le k_j\le L-1}
r_i\big(p_i,w_{k_{i+1}},\ldots,w_{k_n}\big)\,
\Pr_{\bmv_{i+1;n} \sim \calF(\cdot|v_i)}\!\Big[\bigcap_{j=i+1}^n W_j(k_j)\Big].
\end{aligned}
\end{equation}

\vspace{0.2cm}
\noindent\textbf{Bounding the Probability between $\calF$ and $\hat{\calF}$.}
We use log-Lipschitzness (\Cref{asp:lipschitz}) to relate the probability under the continuous distribution $\calF$ and the probability under the discrete distribution $\hat{\calF}$.
Since in event $W_j(k_j)$ we have $v_j \in [w_{k_j}, w_{k_j+1})$, we have $|v_j -w_{k_j}| \le \varepsilon_0$. Applying the log-Lipschitzness condition across all players $(i+1:n)$, we can derive:
\begin{equation}
\label{eq:cell-prob-bound-upper}
\begin{aligned}
&\Pr_{\bmv_{i+1:n}\sim\calF(\cdot|v_i)}\left[\bigcap_\jggei W_j(k_j)\right]
= \iint_{\bmv_{i+1:n}\in \times_\jggei [w_{k_j},w_{k_j+1})} \prod_\jggei f_j(v_j|v_{j-1}) \dd \bmv\\
=& \iint_{\bmv_{i+1:n}\in \times_\jggei [w_{k_j},w_{k_j+1})} \prod_\jggei f_j(v_j|w_{k_{j-1}} + v_{j-1} - w_{k_{j-1}}) \dd \bmv\\
\le& \iint_{\bmv_{i+1:n}\in \times_\jggei [w_{k_j},w_{k_j+1})} \prod_\jggei f_j(v_j|w_{k_{j-1}}) \cdot \exp(K_0\cdot |v_{j-1} - w_{k_{j-1}}|) \dd \bmv\\
\le& \iint_{\bmv_{i+1:n}\in \times_\jggei [w_{k_j},w_{k_j+1})} \prod_\jggei f_j(v_j|w_{k_{j-1}}) \cdot \exp(K_0\cdot\varepsilon_0) \dd \bmv\\
\le& e^{nK_0\varepsilon_0} \cdot \prod_\jggei \int_{v_j\in[w_{k_j},w_{k_j+1})} f_i(v_j|w_{k_{j-1}}) \dd v_j\\
=& e^{nK_0\varepsilon_0} \cdot \prod_\jggei q_{j,k_{j-1},k_j}
= e^{nK_0\varepsilon_0} \cdot \Pr_{\bmv_{i+1:n}\sim \hat{\calF}(\cdot|v_i = w_{k_i})}[\bigcap_\jggei \hat{W}_j(k_j)]
\end{aligned}
\end{equation}

The symmetric inequality from~\Cref{asp:lipschitz} yields a similar lower bound:
\begin{equation}\label{eq:cell-prob-bound-lower}
\Pr_{\bmv_{i+1:n}\sim\calF(\cdot|v_i)}\Big[\bigcap_{j=i+1}^n W_j(k_j)\Big]
~\ge~ e^{-nK_0\varepsilon_0}\, \Pr_{\bmv_{i+1:n}\sim \hat{\calF}(\cdot|v_i = w_{k_i})}\Big[\bigcap_{j=i+1}^n \hat W_j(k_j)\Big].
\end{equation}

Next, we can bound the gain function $g_i(v_i,p_i)$ in $\calF$ relative to the same gain function in $\hat{\calF}$.

\paragraph{Upper Bound on $\ts^\dagger_i(v_i)$.}
Fix $v_i\in [w_{k_i},w_{k_i+1})$.
Let $p'_i \in \mathcal{P}_i$ be player $i$'s deviation in price.
By former discussions, it's without loss of generality to let $p'_i = u^*(i+1,\ell')$ for some $\ell'$ or $p'_i=0$ if $u^*(i+1,\ell')$s are negative for all $\ell'$, since there is always some optimal deviation within $\{0\} \cup \{u^*(i+1,\ell')\}_{0\le \ell' \le L-1}$.
Using $v_i \le w_{k_i} + \varepsilon_0$ and the upper probability bounds \eqref{eq:cell-sum}, \eqref{eq:cell-prob-bound-upper}, and note that $r_i(\cdot)$ is always non-negative, we have
\begin{equation}\label{eq:upper-final}
\begin{aligned}
&g_i(v_i,p'_i) \nonumber\\
&\le w_{k_i} + \varepsilon_0 + \sum_{k_{i+1:n}} r_i\big(p'_i,w_{k_{i+1:n}}\big)\, e^{nK_0\varepsilon_0}
\,\Pr_{\bmv_{i+1:n} \sim \hat{\calF}(\cdot|v_i = w_{k_i})}\!\Big[\bigcap_{j=i+1}^n \hat W_j(k_j)\Big]\nonumber\\
&= w_{k_i} + \varepsilon_0 + e^{nK_0\varepsilon_0} \left( \sum_{k_{i+1:n}} r_i\big(p'_i,w_{k_{i+1:n}}\big)\, 
\,\Pr_{\bmv_{i+1:n} \sim \hat{\calF}(\cdot|v_i = w_{k_i})}\!\Big[\bigcap_{j=i+1}^n \hat W_j(k_j)\Big] \right)\nonumber\\
& = w_{k_i} + \varepsilon_0 + e^{nK_0\varepsilon_0} (u(i,k_i,\ell') - w_{k_i})
\\
&\le w_{k_i} + \varepsilon_0 + e^{nK_0\varepsilon_0}\cdot (u^*(i,k_i) -w_{k_i}).
\\
&\le (1-e^{nK_0\varepsilon_0})L_0  + \varepsilon_0 + e^{nK_0\varepsilon_0}\cdot u^*(i,k_i)
\end{aligned}
\end{equation}

The second equality holds because the term in the parenthesis (plus $w_{k_i}$) is exactly the gain function of player $i$ with valuation distribution $\hat{\calF}$, player $i$'s valuation $v_i = w_{k_i}$ and price $p'_i$, \ie, $u(i,k_i,\ell')$. 
The second inequality holds because $u^*(i, k_i)$ is the maximum of $u(i,k_i,\ell')$ over all $\ell'$. The last inequality holds because $L_0 \le w_{k_i}$.

Next, we know that
\begin{equation}
\label{eq:prf:thm:approx-algo:1}
s^\dagger_i(v_i) \coloneqq \max_{p'_i}\quad g_i(v_i,p'_i) ~\le~ (1-e^{nK_0\varepsilon_0})L_0  + \varepsilon_0 + e^{nK_0\varepsilon_0}\cdot u^*(i,k_i)
\end{equation}

\paragraph{Lower Bound on $g_i(v_i,\tp^*_i(v_i))$.}
For the lifted price $\tp_i^*(v_i)$ ($\tp_i^*(v_i)=u^*(i{+}1,l^*(i,k_i))$ if $l^*(i,k_i) \ne -1$ and $\tp_i^*(v_i)=0$ otherwise), we use $v_i\ge w_{k_i}$ and the lower bound \eqref{eq:cell-prob-bound-lower}:
\begin{equation}\label{eq:lower-final}
\begin{aligned}
&g_i\!\big(v_i,\tp_i^*(v_i)\big) \nonumber\\
&\ge w_{k_i} + \sum_{k_{i+1:n}} r_i\big(\tp_i^*(v_i),w_{k_{i+1:n}}\big)\,
e^{- nK_0\varepsilon_0}\,
\Pr_{\bmv_{i+1:n} \sim \hat{\calF}(\cdot|v_i = w_{k_i})}\!\Big[\bigcap_{j=i+1}^n \hat W_j(k_j)\Big]\nonumber\\
&= w_{k_i} + e^{- nK_0\varepsilon_0}
\left( \sum_{k_{i+1:n}} r_i\big(\tp_i^*(v_i),w_{k_{i+1:n}}\big)\,
\Pr_{\bmv_{i+1:n} \sim \hat{\calF}(\cdot|v_i = w_{k_i})}\!\Big[\bigcap_{j=i+1}^n \hat W_j(k_j)\Big]\right)\nonumber\\
&= w_{k_i} + e^{-nK_0\varepsilon_0}
(u^*(i,k_i) - w_{k_i})\nonumber\\
&\ge (1-e^{-nK_0\varepsilon_0}) L_0 + e^{-nK_0\varepsilon_0}
u^*(i,k_i)
\end{aligned}
\end{equation}

The second equality holds by definition, because the term in the parenthesis (plus $w_{k_i}$) is the gain function of player $i$ with valuation distribution $\hat{\calF}$, player $i$'s valuation $v_i = w_{k_i}$ and optimal price $p^*(i,k_i)$. The last inequality holds because $L_0 \le w_{k_i}$.

\paragraph{Approximation Bound for $H_0 < 0$.}

We first consider $H_0 < 0$, which is roughly an uninterested corner case.
In the game where $v_i$ is guaranteed negative for all $i$, a strategy profile $(\tbmp,\tbms)$ where $\tp_i(v_i)=0$ and $\ts_i(v_i) = v_i$ is a straightforward sequential equilibrium. The approximate equilibrium output by \cref{alg:approx} will return $\ts^*_i(v_i) = w_{k(v_i)}$ and $\tp^*_i(v_i)=0$. It's straightforward to see that as long as $\varepsilon_0 \le \varepsilon$, then $\tbmp^*,\tbms^*$ output by \cref{alg:approx} constitutes an $\varepsilon$-approximate sequential equilibrium.

\paragraph{Approximation Bound for $H_0 \ge 0$.}
Combining the bounds \eqref{eq:upper-final} and \eqref{eq:lower-final}, and noting that $\ts_i^*(v_i)=u^*(i,k_i)$, we have following inequalities:
\begin{equation}\label{eq:prf:thm:approx-algo:2}
\begin{aligned}
\ts_i^*(v_i) \ge& e^{-nK_0\varepsilon_0}\cdot \Big(\ts^\dagger_i(v_i) - (1-e^{nK_0\varepsilon_0})L_0 - \varepsilon_0\Big)
\\
\ts^*_i(v_i) \le& e^{nK_0\varepsilon_0}\cdot \Big(\ts^\dagger_i(v_i) - (1-e^{-nK_0\varepsilon_0})L_0\Big)
\\
g_i(v_i,\tp_i^*(v_i)) \ge& e^{-2nK_0\varepsilon_0}\cdot \Big(\ts^\dagger_i(v_i) - (1-e^{2nK_0\varepsilon_0})L_0 - \varepsilon_0\Big)
\end{aligned}
\end{equation}
and consequently,
\begin{equation}\label{eq:prf:thm:approx-algo:3}
\begin{aligned}
\ts_i^*(v_i) - \ts^\dagger_i(v_i) \ge& (e^{-nK_0\varepsilon_0}-1)\ts^\dagger_i(v_i) - (e^{-nK_0\varepsilon_0}-1)L_0 - e^{-nK_0\varepsilon_0} \varepsilon_0
\\
\ts^*_i(v_i) - \ts^\dagger_i(v_i) \le& (e^{nK_0\varepsilon_0}-1)\ts^\dagger_i(v_i) - (e^{nK_0\varepsilon_0}-1)L_0
\\
g_i(v_i,\tp_i^*(v_i)) - \ts^\dagger_i(v_i) \ge& (e^{-2nK_0\varepsilon_0}-1)\ts^\dagger_i(v_i) - (e^{-2nK_0\varepsilon_0}-1)L_0 - e^{-2nK_0\varepsilon_0} \varepsilon_0
\end{aligned}
\end{equation}

Recall the requirement for $\varepsilon$-approximate sequential equilibrium for strategy profile $(\tbmp^*,\tbms^*)$:

\begin{itemize}[left=0em]
\item $g_i(v_i,\tp_i(v_i)) - \ts^\dagger_i(v_i) \ge - \varepsilon \cdot \max\{1, \ts^\dagger_i(v_i)\}$, 
\item $|\ts_i(v_i) - \ts^\dagger_i(v_i)| \le \varepsilon \cdot \max\{1, \ts^\dagger_i(v_i)\}$, 
\end{itemize}

By regarding $\ts^\dagger_i(v_i)$ as a variable, the approximation requirement $\varepsilon\cdot \max\{1, \ts^\dagger_i(v_i)\}$ is point-wise maximum of two linear functions, and the bound we already have in \cref{eq:prf:thm:approx-algo:3} is a linear function. A sufficient condition for $\varepsilon$-approximation is to check that our linear bound is below the requirement bound for any $\ts^\dagger_i(v_i)$. To do so, we need to guarantee that the slope of linear function in the bound lies between the slope of two linear functions in the approximation requirements, and the value of our bound in $\ts^\dagger_i(v_i)$ is below the value of the approximation requirements in $\ts^\dagger_i(v_i)$. This leads to following inequalities:
\begin{equation}\label{eq:prf:thm:approx-algo:4}
\begin{aligned}
& -\varepsilon \le e^{-n K_0 \varepsilon_0} - 1 \le 0,
& & (e^{-nK_0\varepsilon_0}-1)(1-L_0) - e^{-nK_0\varepsilon_0} \varepsilon_0 \ge -\varepsilon
\\
& 0 \le e^{n K_0 \varepsilon_0} - 1 \le \varepsilon,
& & (e^{nK_0\varepsilon_0}-1)(1-L_0) \le \varepsilon
\\
& -\varepsilon \le e^{-2n K_0 \varepsilon_0} - 1 \le 0,
& & (e^{-2nK_0\varepsilon_0}-1)(1-L_0) - e^{-2nK_0\varepsilon_0} \varepsilon_0 \ge -\varepsilon
\end{aligned}
\end{equation}

By choosing $\varepsilon_0 = \frac{\min(\varepsilon,1)}{4(1 + n K_0)(1 + \max(0,-L_0))}$, above inequalities hold simultaneously.
This establishes that the strategy profile $(\tbmp^*,\tbms^*)$ is an $\varepsilon$-approximate sequential equilibrium.


\vspace{0.2cm}
\noindent\textbf{Complexity Analysis}\label{subsubsec:fptas-complexity}
With the chosen $\varepsilon_0$, the grid size $L$ is:
\[
L ~=~ \Big\lfloor \frac{4(1 + n K_0)(1 + \max(0,-L_0))(H_0 - L_0)}{\min(\varepsilon,1)} \Big\rfloor + 1,
\]
It's clear to see that $L$ is polynomial in $n,K_0,\varepsilon^{-1},H_0 - L_0$, and the algorithm for exact equilibrium runs in $O(n^2L^4)$ time (\Cref{thm:algo-exact}). Therefore, the overall running time of \cref{alg:approx} is $\poly\big(n,K_0,\varepsilon^{-1},H_0 - L_0\big)$.

\end{proof}

\subsection{Proof of \texorpdfstring{\cref{thm:partial}}{}}
\thmPartial*
\begin{proof}\label{prf:thm:partial}

We start with a technical lemma that establishes the property of the equilibrium.

\begin{restatable}{lemma}{lemSequentialInvariance}
\label{lem:sequential-invariance}
Let $(\tbmp^*,\tbms^*)$ be the sequential equilibrium under mechanism $\bmt$. Fix player index $i$ and valuation $v$. If $t_{j,i}=0$ for all $j\ne i+1$, then
$\tp^*_i(v)=\hat p^*_{i+1}\cdot \ts^*_{i+1}(1)\cdot v$,
for some constant $\hat p^*_{i+1}$ determined solely by $F_{i+1}$.
\end{restatable}

\noindent The proof of \cref{lem:sequential-invariance} is only technical and deferred to \cref{sec:omitted-proofs}. We now focus on proving the main theorem.

\paragraph{Proving (2).}

By condition (a), the profit reallocation mechanism $t^k_{j,i'}$ for all players $i' \ge i+1$ in the subgame starting at player $i+1$ are identical between $k=1,2$.
Since the equilibrium in the subgame is determined solely by the model representation and profit reallocation mechanism, the equilibrium strategies must coincide, as $\argmax$ is uniquely defined by some tie-breaking rule. Therefore, we must have $\tp^{2,*}_j=\tp^{1,*}_j$ and $\ts^{2,*}_j=\ts^{1,*}_j$ for all $j>i$. This directly proves (2) for $j>i+1$.


For trades $j \le i$, by condition (d), player $j-1$ only receives one-shot profit from trade $j$ ($t^k_{l,j-1}=0$ for $l>j$). This satisfies the condition for \cref{lem:sequential-invariance}. Thus, the price set by seller $j-1$ satisfies $\tp^{k,*}_{j-1}(v_{j-1}) = \hat p^*_j \cdot \ts^{k,*}_{j}(1) \cdot v_{j-1}$, where $\hat p^*_j$ is a constant determined solely by $F_{j+1}$.

The equilibrium trade event then becomes,
\begin{equation*}
\begin{aligned}
E^k_i
= \{\bmv: \ts^{k,*}_j(v_{j}) \ge \tp^{k,*}_{j-1}(v_{j-1})\}
= \{\bmv: v_j \cdot \ts^{k,*}_j(1) \ge v_{j-1} \cdot \hat p^*_i \cdot \ts^{k,*}_j(1) \}
= \{\bmv: v_j \ge v_{j-1}\cdot \hat p^*_i \}.
\end{aligned}
\end{equation*}
which is independent of the mechanism $k \in \{1, 2\}$. Therefore, the equilibrium trade event $E^k_i$ is constant between $k=1,2$ for $j \le i$. This completes the proof of (2).


\paragraph{Proving (1).}

Since the threshold strategy $\ts^{k,*}_{i+1}(v_{i+1})$ is the same for $k=1,2$, proving $\tp^{2,*}_i(v)\ \le\ \tp^{1,*}_i(v)$ (Statement (1)) is sufficient to prove $E^1_{i+1} \subseteq E^2_{i+1}$.
Therefore, we focus on proving $\tp^{2,*}_i(v)\ \le\ \tp^{1,*}_i(v)$.

Let $g^k_i(v_i,p_i)$ be the gain function of seller $i$ for setting price $p_i$ under mechanism $\bmt^k$ (\cref{eq:lem:simplify-buy:3}), we make below decomposition:
\begin{equation}
\begin{aligned}
g^k_i(v_i,p_i)\coloneqq =v_i+r^k_i(v_i,p_i)+d^k_i(v_i,p_i),\quad
r^k_i(v_i,p_i)\coloneqq \bbE\!\left[t^k_{i+1,i}\,p_i\,\one\!\{\ts^{k,*}_{i+1}(v_{i+1})\ge p_i\}\right],\\
d^k_i(v_i,p_i)\coloneqq \bbE \Big[\sum_{i<j<n} t^k_{j+1,i}\,\tp^{k,*}_j(v_j)\,\one\!\{\ts^{k,*}_{i+1}(v_{i+1})\ge p_i\}
\prod_{i<j'\le j}\one\!\{\ts^{k,*}_{j'+1}(v_{j'+1})\ge \tp^{k,*}_{j'}(v_{j'})\}\Big].
\end{aligned}
\end{equation}
where $r^k_i$ is the one-shot profit and $d^k_i$ is the future reallocated profit.
The gain function $g^k_i(v_i,p_i)$ becomes
\begin{equation}
g^k_i(v_i,p_i)=v_i+r^k_i(v_i,p_i)+d^k_i(v_i,p_i).
\end{equation}
We focus on the change in the gain function $\Delta g_i(p) \coloneqq g^2_i(v_i,p) - g^1_i(v_i,p)$ when the mechanism moves from $\bmt^1$ to $\bmt^2$.
Let $p^k = \tp^{k,*}_i(v_i)$ be the optimal price for mechanism $\bmt^k$. 
By conditions (b) and (c):
\begin{itemize}[left=0em]\itemsep0.2em
\item Let $\Delta r_i(v_i,p) \coloneqq r^2_i(v_i,p) - r^1_i(v_i,p)$. Then, for any $p, p'$,
\begin{equation}
r^2_i(v_i,p) \le r^2_i(v_i,p')\quad \Longrightarrow\quad \Delta r_i(v_i,p) \ge \Delta r_i(v_i,p'),
\end{equation}
because $\ts^{1,*}_{i+1} = \ts^{2,*}_{i+1}$ and consequently $r^2_i(v_i,p) = \frac{t^2_{i+1,i}}{t^1_{i+1,i}} r^1_i(v_i,p)$ with coefficient $\frac{t^2_{i+1,i}}{t^1_{i+1,i}} \le 1$.

\item The difference $\Delta d_i(v_i,p):=d^2_i(v_i,p)-d^1_i(v_i,p)$ and $d^k_i(v_i,p), k\in\{1,2\}$ are \emph{weakly decreasing} in $p$, because higher $p$ weakly reduces the event $\{\ts^{k,*}_{i+1}(v_{i+1}) \ge p\}$ and $t^2_{j+1,i} \ge t^1_{j+1,i}\ge 0$ for all $j\ge i+1$.
\end{itemize}

Now, suppose, for contradiction, that $p^2>p^1$.
By the optimality of $p^2$ under $\bmt^2$, we have $g^2_i(v_i,p^2)\ge g^2_i(v_i,p^1)$.
Since $d^2_i$ is weakly decreasing in $p$, we must have $d^2_i(v_i,p^2) \le d^2_i(v_i,p^1)$.
This implies that:
\begin{equation}\label{eq:thm:partial:1}
\begin{aligned}
g^2_i(v_i,p^2)\ge g^2_i(v_i,p^1),\; d^2_i(v_i,p^2) \le d^2_i(v_i,p^1)
\Rightarrow r^2_i(v_i,p^2) \ge r^2_i(v_i,p^1)
\Rightarrow \Delta r_i(v_i,p^2) \le \Delta r_i(v_i,p^1).
\end{aligned}
\end{equation}
We now analyze the term $S \coloneqq \left[g^1_i(v_i,p^2) - g^1_i(v_i,p^1)\right] - \left[g^2_i(v_i,p^2) - g^2_i(v_i,p^1)\right]$.
\begin{equation*}
\begin{aligned}
S =& r^1_i(v_i,p^2) - r^1_i(v_i,p^1)
+ d^1_i(v_i,p^2) - d^1_i(v_i,p^1)
- \big(r^2_i(v_i,p^2) - r^2_i(v_i,p^1)
+ d^2_i(v_i,p^2) - d^2_i(v_i,p^1)\big)
\\
=& \Delta r_i(v_i,p^1) - \Delta r_i(v_i,p^2) + \Delta d_i(v_i,p^1) - \Delta d_i(v_i,p^2) \ge 0
\end{aligned}
\end{equation*}
where $\Delta r_i(v_i,p^1) - \Delta r_i(v_i,p^2) \ge 0$ is because \cref{eq:thm:partial:1} and $\Delta d_i(v_i,p^1) - \Delta d_i(v_i,p^2)\ge 0$ is because $\Delta d_i(v_i,p)$ weakly decreases in $p$ and $p^2 > p^1$. Then, we conclude that
$g^1_i(v_i,p^2) - g^1_i(v_i,p^1) \ge 0$,
which indicates $g^1_i(v_i,p^2) = g^1_i(v_i,p^1)$ by optimality of $p^1$, and $g^2_i(v_i,p^2) = g^2_i(v_i,p^1)$ by $S\ge 0$. But this indicates $p^1 = p^2$ because the tie-breaking rule does not rely on $\bmt$, which contradicts the assumption $p^2 > p^1$ and proves (1).
\end{proof}

\subsection{Proof of \texorpdfstring{\cref{lem:sequential-invariance}}{}}
\label{prf:lem:sequential-invariance}

We will actually prove a more general lemma:

\begin{restatable}{lemma}{lemSequentialInvarianceTwo}
\label{lem:sequential-invariance-2}
Let $(\tbmp^*,\tbms^*)$ be the sequential equilibrium under mechanism $\bmt$. Then, for all player index $i$, $\alpha>0$ and valuation $v$:
\begin{enumerate}\itemsep0.25em
\item[(1)] $\ts^*_i(\alpha v)=\alpha\,\ts^*_i(v)$;
(2) $\tp^*_i(\alpha v)=\alpha\,\tp^*_i(v)$
\item[(3)] If, in addition, $t_{j,i}=0$ for all $j\ne i+1$, then
\[
\tp^*_i(v)=\hat p^*_{i+1}\cdot \ts^*_{i+1}(1)\cdot v,
\]
for some constant $\hat p^*_{i+1}$ determined solely by $F_{i+1}$.
\end{enumerate}
\end{restatable}

We first provide some lemmas (\cref{lem:price-invariance,lem:partial-ui,lem:partial-ui:2}), then prove the main lemma (\cref{lem:sequential-invariance-2}).

\begin{lemma}
\label{lem:price-invariance}
Let $v_i>0$, $\hat{p}_i^*(v_i) = \argmax_{p} G_i(p|v_i)\cdot p$ be the optimal price to sell a data product to some player with distribution $F_i(\cdot|v_i)$, then $\hat{p}_i^*(\alpha v_i) = \alpha \hat{p}_i^*(v_i)$ for $\alpha > 0$.

Let $\hat{\alpha}_i^*(v_i) = G_i(\hat{p}_i^*(v_i)|v_i)$ be the probability that the data product is sold, then $\hat{\alpha}_i^*(v_i) = \hat{\alpha}_i^*$ for some $\hat{\alpha}_i^* > 0$.
\end{lemma}

\begin{proof}[Proof of \cref{lem:price-invariance}]

We first show $\hat{p}_i^*(\alpha v_i) = \alpha \hat{p}_i^*(v_i)$. By simple calculation,
\begin{equation}
\begin{aligned}
& G_i(\hat{p}_i^*(\alpha v_i)|\alpha v_i) \cdot \hat{p}_i^*(\alpha v_i)
\\
=& \max_p\quad G_i(p|\alpha v_i) \cdot p 
\\
=& \max_p\quad G_i(\alpha p|\alpha v_i) \cdot \alpha p \cdots\text{rewrite $p$ with $\alpha p$}
\\
=& \max_p\quad G_i(p|v_i) \cdot \alpha p 
\\
=& G_i(\hat{p}_i^*(v_i)|v_i) \cdot \alpha \hat{p}_i^*(v_i) 
\\
&\cdots \text{Note that $\hat{p}_i^*(v_i)$ is the maximum point by definition}
\\
=& G_i(\alpha \hat{p}_i^*(v_i)|\alpha v_i) \cdot \alpha \hat{p}_i^*(v_i)
\end{aligned}
\end{equation}

By tie-breaking rule, both $\hat{p}_i^*(\alpha v_i)$ and $\alpha \hat{p}_i^*(v_i)$ is the argmax point of $G_i(p|\alpha v_i)\cdot p$.
Therefore, it must have $\hat{p}_i^*(\alpha v_i) = \alpha \hat{p}_i^*(v_i)$.
Lastly, we prove that $\hat{q}_i^*(\alpha v_i) = \hat{q}_i^*(v_i)$. Note that $\hat{q}_i^*(\alpha v_i) = G_i(\hat{p}_i^*(\alpha v_i)|\alpha v_i) = G_i(\alpha \hat{p}_i^*(v_i)|\alpha v_i) = G_i(\hat{p}_i^*(v_i)|v_i) = \hat{q}_i^*(v_i)$. We then complete the proof.
\end{proof}

\begin{lemma}
\label{lem:partial-ui}
Under conditions in \cref{lem:sequential-invariance-2}. Fix player index $i$ and assume \cref{lem:sequential-invariance-2} holds for $j > i$, then we have $g_i(\alpha_i v_i,\alpha_i p_i) = \alpha_i g_i(v_i,p_i)$.
\end{lemma}

\begin{proof}[Proof of \cref{lem:partial-ui}]
\label{prf:lem:partial-ui}
We rewrite the form of $g_i(v_i,p_i)$ below:

\begin{equation}
\label{eq:lem:partial-ui:1}
\begin{aligned}
g_i(v_i,p_i) =& v_i + \bbE_{\bmv_{i+1:n}\sim \calF(\cdot|v_i)} \left[t_{i+1,i} \cdot p_i\cdot \one\{\ts^*_{i+1}(v_{i+1})\ge p_i\}\right.
\\
+& \sum_{i<j<n} t_{j+1,i} \cdot \tp^*_j(v_j)
\cdot \one\{\ts^*_{i+1}(v_{i+1})\ge p_i\}
\cdot \left. \left(\prod_{i< j'\le j} \one\{\ts^*_{j'+1}(v_{j'+1}) > \tp^*_{j'}(v_{j'})\}\right)
\right]
\end{aligned}
\end{equation}
We briefly explain what this form means: the first line consists of the valuation of data and the one-shot profit, the second line is the sum of future reallocated profit from future trade $j$, where the event $\one\{\ts^*_{i+1}(v_{i+1})\ge p_i\}$ represents the success of trade $i$.
Similarly, $\one\{\ts^*_{j'+1}(v_{j'+1}) > \tp^*_{j'}(v_{j'})\}$ is the event that trade $j+1$ succeeds.
Note that player $i$ can receive reallocated profit from trade $j$ if and only if all trades $i+1,\ldots,j-1,j$ succeed, represented by a product of event indicators
$$\one\{\ts^*_{i+1}(v_{i+1})\ge p_i\}\cdot \left(\prod_{i< j'\le j} \one\{\ts^*_{j'+1}(v_{j'+1}) > \tp^*_{j'}(v_{j'})\}\right)$$ representing that all trades in $i+1,...,j$ succeed.

We then transform $\bmv_{i+1:n}\sim \calF(\cdot|v_i)$ to 
\begin{equation}
v_{i+1}\sim F_{i+1}(\cdot|v_i),\quad \bmv_{i+2:n}\sim \calF(\cdot|v_{i+1})
\end{equation}
by \cref{def:markov}.
We merge the $\one\{v_{i+1}\ge p_i\}$ term in \cref{eq:lem:partial-ui:1} and get the follows:

\begin{equation}
\begin{aligned}
& g_i(v_i,p_i) = v_i + 
\bbE_{v_{i+1}\sim F_{i+1}(\cdot|v_i)}
\left[
\one\{v_{i+1} \ge p_i\} \right.
\\
\cdot& 
\left.
\left(
t_{i+1,i}\cdot p_i + 
\bbE_{\bmv_{i+2:n}\sim \calF(\cdot|v_{i+1})}
\left[
\sum_{i<j<n} t_{j+1,i} \cdot \tp^*_j(v_j) \cdot X_{i,j}(\bmv_{i+1:j+1})
\right]
\right)
\right]
\end{aligned}
\end{equation}
where $X_{i,j}(\bmv_{i+1:j+1}) = \prod_{i< j'\le j} \one\{\ts^*_{j'+1}(v_{j'+1}) > \tp^*_{j'}(v_{j'})\}$ represents the event that trade $i+1$ to $j+1$ succeeds.

We define the inner term

\begin{align*}
w_{i+1}(v_{i+1}) \coloneqq \bbE_{\bmv_{i+2:n}\sim \calF(\cdot|v_{i+1})}
\left[
\sum_{i<j<n} t_{j+1,i} \cdot \tp^*_j(v_j) \cdot X_{i,j}(\bmv_{i+1:j+1})
\right]
\end{align*}
then 

\begin{equation}
\begin{aligned}
g_i(v_i,p_i) =& v_i + 
\bbE_{v_{i+1}\sim F_{i+1}(\cdot|v_i)}
\left[
\one\{v_{i+1} \ge p_i\} \right.
\cdot
\left.
( t_{i+1,i}\cdot p_i + w_{i+1}(v_{i+1}) )
\right]
\end{aligned}
\end{equation}

We next show the following lemma holds.
\begin{lemma}
\label{lem:partial-ui:2}
Under the conditions in \cref{lem:partial-ui},
$w_{i+1}(\alpha v_{i+1}) = \alpha w_{i+1}(v_{i+1})$ for $\alpha > 0$.
\end{lemma}

\begin{proof}[Proof of \cref{lem:partial-ui:2}]
\label{prf:lem:partial-ui:2}

We observe that $X_{i,j}(\bmv_{i+1:j+1})$ is invariant under linear scaling, \ie, $X_{i,j}(\bmv_{i+1:j+1}) = X_{i,j}(\alpha v_{i+1},...,\alpha v_{j+1})$. This is because each inner term in $X_{i,j}$, $\one\{\ts^*_{j'+1}(\alpha v_{j'+1}) \ge \tp^*_{j'}(\alpha v_{j'})\} = \one\{\alpha \ts^*_{j'+1}(v_{j'+1}) \ge \alpha \tp^*_{j'}(v_{j'})\} = \one\{\ts^*_{j'+1}(v_{j'+1}) \ge \tp^*_{j'}(v_{j'})\}$. The first equality holds because the inductive assumption in \cref{lem:partial-ui}.

By \cref{asp:homo}, we know that $\calF(\bmv_{i+2:n}|v_{i+1}) = F_{i+2}(v_{i+2}|v_{i+1})\cdot...\cdot F_n(v_n|v_{n-1}) = F_{i+2}(\alpha v_{i+2}|\alpha v_{i+1})\cdot...\cdot F_n(\alpha v_n|\alpha v_{n-1})$. Then,
\begin{equation}
\label{eq:lem:partial-ui:2}
\begin{aligned}
w_{i+1}(\alpha v_{i+1}) =& \bbE_{\bmv_{i+2:n}\sim \calF(\cdot|\alpha v_{i+1})}
\left[
\sum_{i<j<n} t_{j+1,i} \cdot \tp^*_j(v_j) \cdot X_{i,j}(\alpha v_{i+1}, v_{i+2:j+1})
\right]
\\
=& \bbE_{(\alpha v_{i+2},...,\alpha v_n)\sim \calF(\cdot|\alpha v_{i+1})} 
\left[
\sum_{i<j<n} t_{j+1,i} \cdot \tp^*_j(\alpha v_j) \cdot X_{i,j}(\alpha v_{i+1},...,\alpha v_{j+1})
\right]
\\
=& \bbE_{\bmv_{i+2:n}\sim \calF(\cdot|v_{i+1})} 
\left[
\sum_{i<j<n} t_{j+1,i} \cdot \alpha \tp^*_j(v_j) \cdot X_{i,j}(v_{i+1},...,v_{j+1})
\right]
\\
=& \alpha w_{i+1}(v_{i+1})
\end{aligned}
\end{equation}

Then we complete the proof of \cref{lem:partial-ui:2}.

\end{proof}

We then prove the \cref{lem:partial-ui}. By simple calculation

\begin{equation}
\label{eq:lem:partial-ui:3}
\begin{aligned}
&g_i(\alpha v_i,\alpha p_i) 
\\
=& \alpha v_i + 
\bbE_{v_{i+1}\sim F_{i+1}(\cdot|\alpha v_i)}
\left[
\one\{v_{i+1} \ge \alpha p_i\} \cdot
( t_{i+1,i}\cdot \alpha\cdot p_i + w_{i+1}(v_{i+1}) )
\right]
\\
=& \alpha v_i + 
\bbE_{\alpha v_{i+1}\sim F_{i+1}(\cdot|\alpha v_i)}
\left[
\one\{\alpha v_{i+1} \ge \alpha p_i\} \cdot
( t_{i+1,i}\cdot \alpha\cdot p_i + w_{i+1}(\alpha v_{i+1}) )
\right] \cdots\text{let $v_{i+1} \gets \alpha v_{i+1}$.}
\\
=& \alpha v_i + 
\bbE_{v_{i+1}\sim F_{i+1}(\cdot|v_i)}
\left[
\one\{v_{i+1} \ge p_i\} \cdot
( t_{i+1,i}\cdot \alpha\cdot p_i + \alpha w_{i+1}(v_{i+1}) )
\right] \cdots \text{rewriting each term directly}
\\
=& \alpha v_i + \alpha \cdot 
\bbE_{v_{i+1}\sim F_{i+1}(\cdot|v_i)}
\left[
\one\{v_{i+1} \ge p_i\} \cdot
( t_{i+1,i}\cdot p_i + w_{i+1}(v_{i+1}) )
\right]
\\
=& \alpha g_i(v_i,p_i)
\end{aligned}
\end{equation}

This completes the proof of \cref{lem:partial-ui}.


\end{proof}

Now, we begin with the proof of main lemma (\cref{lem:sequential-invariance-2}).

\begin{proof}[Proof of \cref{lem:sequential-invariance-2}]

We will use mathematical induction from the last player to the first player.
Firstly $\ts^*_n(v) = v$ which satisfies (1) directly. Then the distribution of $\ts^*_n(v_n)$ given $v_{n-1}$ is exactly $F_n$. Player $n-1$'s optimal selling strategy $\tp^*_{n-1}(v_{n-1})$ will maximizes $G_n(p|v_{n-1})\cdot p$, which by \cref{lem:price-invariance} satisfies (2).

Now, assume $\ts^*_{i+1},\tp^*_{i+1},...\ts^*_n$ satisfies (1)(2). This satisfies the requirement of \cref{lem:partial-ui} by setting $j=i$.
Then we will prove $\ts^*_i$ satisfies (1) and $\tp^*_i$ satisfies (2) accordingly, by utilizing the conclusion of \cref{lem:partial-ui}.

We first show that how (1)(2) can be derived from \cref{lem:partial-ui}.

\paragraph{$\ts^*_i$ satisfies (1) and $\tp^*_i$ satisfies (2) given \cref{lem:partial-ui}.}

Taking maximum on (a) in \cref{lem:partial-ui}, we get
\begin{align*}
\ts_i^*(v_i) =& \max_{p_i} g_i(\alpha_i v_i, \alpha_i, p_i) 
\\
=& \alpha_i \max_{p_i} g_i(v_i,p_i) = \alpha_i \ts^*_i(v_i)
\end{align*}

For $p^*_i$, notice that
\begin{align*}
\ts^*_i(\alpha_i v_i, \tp^*_i(\alpha_i v_i))
=& \max_{p_i} g_i(\alpha_i v_i, p_i)
\\
=& \max_{p_i} g_i(\alpha_i v_i, \alpha_i, p_i)
\\
=& \alpha_i \max_{p_i} g_i(v_i,p_i)
\\
=& \alpha_i g_i(v_i, \tp^*_i(v_i))
\\
=& g_i(\alpha_i v_i, \alpha_i \tp^*_i(v_i))
\end{align*}

By tie-breaking rule for argmax, we must have $\alpha_i \tp^*_i(v_i) = \tp^*_i(\alpha_i v_i)$. This completes the induction step in \cref{lem:sequential-invariance-2}.

Now, by mathematical induction we complete the proof of \cref{lem:sequential-invariance-2} (a),(b).

For (c), note that $g_i(v_i,p_i)$ degrades to:
\begin{equation}
\begin{aligned}
g_i(v_i,p_i) =& v_i + \bbE_{v_{i+1}\sim F_{i+1}(\cdot|v_i)} 
\left[t_{i+1,i} \cdot p_i\cdot \one\{\ts^*_{i+1}(v_{i+1})\ge p_i\}\right]
\\
=& v_i + t_{i+1,i} \cdot \bbE_{v_{i+1}/v_i \sim F_{i+1}(\cdot|1)} 
\left[p_i\cdot \one\{\ts^*_{i+1}(v_{i+1})\ge p_i\}\right]
\\
=& v_i + t_{i+1,i} \cdot \bbE_{v_{i+1} \sim F_{i+1}(\cdot|1)} 
\left[p_i\cdot \one\{\ts^*_{i+1}(v_{i+1}\cdot v_i)\ge p_i\}\right]
\cdots\text{by setting $v_{i+1} \gets v_i \cdot v_{i+1}$}
\\
=& v_i + t_{i+1,i} \cdot \bbE_{v_{i+1} \sim F_{i+1}(\cdot|1)} 
\left[p_i\cdot \one\{v_{i+1} \ge \frac{p_i}{(\ts^*_{i+1}(1)\cdot v_i )}\}\right] \cdots\text{$\ts^*_{i+1}(v_{i+1}\cdot v_i) = v_{i+1}\cdot v_i\cdot\ts^*_{i+1}(1)$}
\\
=& v_i + t_{i+1,i}\cdot \ts^*_{i+1}(1)\cdot v_i  \cdot \bbE_{v_{i+1} \sim F_{i+1}(\cdot|1)} 
\left[\frac{p_i}{(\ts^*_{i+1}(1)\cdot v_i )}\cdot \one\{v_{i+1} \ge \frac{p_i}{(\ts^*_{i+1}(1)\cdot v_i )}\}\right]
\end{aligned}
\end{equation}

By tie-breaking rule, it's clear that $\frac{\tp^*_i(v_i)}{(\ts^*_{i+1}(1)\cdot v_i )}$ is a constant that depends only on $F_{i+1}$, say, $\hat{p}^*_{i+1}$.
Therefore, $\tp^*_i(v_i) = \hat{p}^*_{i+1}\cdot \ts^*_{i+1}(1)\cdot v_i$.
This ends the proof of \cref{lem:sequential-invariance-2}.


\end{proof}

\subsection{Proof of \texorpdfstring{\cref{thm:global}}{}}
\thmGlobal*

\begin{proof}
\label{prf:thm:global}
The proof proceeds by constructing a sequence of mechanisms that progressively transitions from the baseline mechanism $\bmt$ to the general mechanism $\bmt'$.
We define a chain of $n$ mechanisms, $\bmt^n=\bmt, \bmt^{n-1}, \ldots, \bmt^1, \bmt^0=\bmt'$, by setting the $k$-th mechanism $\bmt^k$ such that the profit reallocations for players $j\geq k$ matches $\bmt'$, while the profit reallocations for players $j < k$ match the baseline mechanism $\bmt$:
\[
t^k_{ij} \;:=\;
\begin{cases}
t'_{ij}, & \text{if } j\ge k,\\
\one\{i=j+1\}, & \text{otherwise}.
\end{cases}
\]
For each $k\in\{n,\ldots,1\}$, the adjacent pair $(\bmt^k,\bmt^{k-1})$ differs only in profit reallocation to player $k-1$. One can verify that setting $\bmt^1 \gets \bmt^k,\bmt^2 \gets \bmt^{k-1}$ in \cref{thm:partial} satisfies the conditions of \cref{thm:partial} with $i=k-1$. Therefore, \cref{thm:partial} implies that, by denoting $E^k_i \coloneqq \{\bmv\in\calV: \one\{\ts^{k,*}_i(v_i) \ge \tp^{k,*}_{i-1}(v_{i-1})\}\}$ where $(\tbmp^k,\tbms^k)$ is the equilibrium under mechanism $\bmt^k$, we have

\begin{equation}
E^k_i \subseteq E^{k-1}_i
\end{equation}

Consequently, $E_i = E^n_i \subseteq E^0_i = E'_i$, which completes the proof.



\end{proof}

\section{Experimental Details}
\label{sec:omitted_exp}

We describe how we use SciPy to learn an optimal profit reallocation mechanism. All experiments were run on CPUs with a single MacBook.

Firstly, for each player index $i$, we apply a softmax to a real vector $x_{i,0:i-1} \in \bbR^i$ so that $t_{i,0:i-1} = \mathrm{softmax}(x_{i,0:i-1})$ automatically satisfies $\sum_{j} t_{ij} = 1$. This reparameterization turns the constrained problem into an unconstrained one with variables $\bmx=\{x_{ij}\}_{0\le j<i}$.

The model representation, \cref{alg:exact}, and the zero-th order optimization are all deterministic for fixed inputs: the only stochastic ingredient in the welfare objective is the estimation of expected social welfare w.r.t. $\bmv\sim\calF$.
For each model, we draw $M=10{,}000$ independent valuation profiles and average social welfare within these profiles; this yields a high-precision estimate of expected social welfare.

Then, given $\bmx$, we map it to a profit reallocation mechanism $\bmt$ and call \cref{alg:exact} to obtain the sequential equilibrium. We evaluate that equilibrium on the same $M$ profiles to approximate the expected social welfare. The pipeline from $\bmx$ to expected social welfare is denoted as $\SW(\bmx)$, which is the objective function. We then use \emph{differential evolution}, a zeroth-order algorithm 
packaged in scipy, to find the maximum point of $\SW(\bmx)$. We use $50$ iterations, population size $5$, and initialize from the AVE mechanism.
The running time is approximately one minute for model with $n=5$.

\section{Potential Social Impact}
\label{sec:impact}
By incorporating profit reallocation mechanisms to real data markets, data trading is likely to become more active, consequently facilitating data circulation and data utilization.
As is suggested by experiments, The interests of a small group of participants may be affected, especially those upstream (\eg, data owner) in the data transaction chain.
Additionally, this work does not directly resolve the risk (\eg, data leakage, data abuse) arose from inappropriate data usage. The promotion of data trading may also increase such risks. 

\end{document}